\documentclass[a4pape,reqno,11pt,oneside]{amsart}
\usepackage[foot]{amsaddr}
\usepackage{amssymb}
\usepackage{amsmath}
\usepackage{amsthm}
\usepackage{a4wide}
\usepackage{cite}
\usepackage{tikz}
\usepackage{stix}
\usepackage{color}
\usepackage{graphpap}
\usepackage{epsfig}
\usepackage{psfrag}
\usepackage{graphicx}
\usepackage{subfigure}
\usepackage{tikz-network}
\usetikzlibrary{arrows,snakes,backgrounds,automata,trees,shapes}
\usepackage{multirow}
\usepackage{threeparttable}
\usepackage{float}
\usepackage{rotating}
\usepackage{bbm}
\usepackage{comment}
\usepackage{url}

\newtheorem{theorem}{Theorem}[section]
\newtheorem{corollary}{Corollary}[section]

\newtheorem{prop}{Proposition}[section]

\theoremstyle{definition}

\theoremstyle{remark}

\begin{document}


%********************************************************

\title[How far can a rumor travel without shortcuts?]{How far can a rumor travel without shortcuts?}

\author[Ana C. D\'iaz Bacca]{Ana C. D\'iaz Bacca$^1$}
\address{$^1$ Ana C. D\'iaz Bacca and Catalina M. R\'ua-Alvarez: Departamento de matem\'aticas y estad\'istica, Universidad de Nari\~no, Cll. 18 No. 50 - 2, Ciudadela Universitaria Torobajo, Pasto - Nariño, Colombia. E-mail addresses: \url{acdiaz@udenar.edu.co}, \url{catalina.rua@udenar.edu.co}}

\author[Pablo M. Rodriguez ]{Pablo M. Rodriguez$^2$}
\address{$^2$ Pablo M. Rodriguez (Corresponding author): Centro de Ci\^encias Exatas e da Natureza, Universidade Federal de Pernambuco, Av. Prof. Moraes Rego, 1235 - Cidade Universit\'aria - Recife - PE, Brazil. E-mail address: \url{pablo@de.ufpe.br}}

\author[Catalina M. R\'ua-Alvarez]{Catalina M. R\'ua-Alvarez$^1$}
%\address{Catalina M. R\'ua-Alvarez: Departamento de matem\'aticas y estad\'istica, Universidad de Nari\~no, Cll. 18 No. 50 - 2, Ciudadela Universitaria Torobajo, Pasto - Nariño, Colombia.}
%\email{catalina.rua@udenar.edu.co}

%% Abstract
\begin{abstract}
We consider a rumor model in which the network is divided into three classes of agents: ignorant, spreader, and stifler. A spreader transmits the rumor to each of its ignorant neighbors at rate one, and at the same rate, it becomes a stifler after interacting with other spreaders or stiflers. The overall process is described by a continuous-time Markov chain that represents the state of each node at any given time. The underlying network is a ring lattice with $n$ nodes, where each node is connected to its $2k$ nearest neighbors. This structure has often been used as the foundation for small-world network models, which are typically generated by rewiring or adding edges to introduce shortcuts. It is well known that when a rumor process takes place on such modified networks, the system undergoes a transition between localization and propagation at a finite mean degree. This transition illustrates the strong influence of shortcuts on the spreading of information. In this work, we adopt a complementary perspective by focusing on the rumor process within the pure ring lattice, without adding any shortcuts. Our aim is to show that even in this simplified setting, the model can exhibit behavior regarding the proportion of nodes reached by the rumor that is comparable to what is observed in homogeneously mixed populations. To this end, we identify the value of $k$ as a function of $n$ for which this behavior emerges and demonstrate that it scales as $\log n$. Our conclusions are drawn from the analysis of contrasting examples and from a broader examination of the general case through numerical simulations.
\end{abstract}

\maketitle

\section{Introduction}

\subsection{Rumor spreading on networks: a mathematical model} We consider a mathematical (probabilistic) model for the representation of rumor spreading in a population. The first references to this type of model can be found in \cite{dg,daley_nature,kendall, MT}. In particular,  we focus on the Maki-Thompson rumor model, which is defined as a continuous-time Markov chain. It is assumed that the population may be represented by a network in which the nodes represent agents and the edges represent possible interactions. Furthermore, the population is divided into three categories of agents: those who are unaware of the information (ignorants), those who disseminate it (spreaders), and those who are aware of the information but do not (stiflers). A spreader transmits the rumor to one of its ignorant neighbors at a rate of one. At the same rate, the spreader becomes stifler after contact with other spreaders or stiflers. See Figure \ref{FIG:transitionsMT}. In other words, interactions occur in accordance with the marks of independent Poisson processes. This model was introduced in \cite{MT} to describe the dissemination of information within a population. This is a simplified version of the Daley-Kendall rumor model, which was proposed as an alternative to the well-known SIR epidemic model for representing the spread of information. For further details, see \cite{dg,daley_nature}. The Maki-Thompson rumor model may be formally defined on a network with $n$ nodes as a continuous-time Markov process $\{\eta_t\}_{t\geq 0}$ with state space $\mathcal{S}=\{0,1,2\}^{[n]}$, where $[n]:=\{1,2,\ldots,n\}$ denotes the set of nodes, that is, at time $t$ the state of the process is a function $\eta_t: [n] \longrightarrow \{0,1,2\}$. We assume that each node $v$ of the network represents an agent, which is said to be ignorant if $\eta_t(v)=0,$ spreader if $\eta_t(v)=1$, and stifler if $\eta_t(v)=2,$ at time $t$. Then, if the system is in the configuration $\eta \in \mathcal{S},$ the state of node $v$ changes according to the following transition rates:

\begin{equation}\label{rates}
\begin{array}{rclc}
&\text{transition} &&\text{rate} \\[0.1cm]
0 & \rightarrow & 1, & \hspace{.5cm}  n_{1}(v,\eta),\\[.2cm]
1 & \rightarrow & 2, & \hspace{.5cm}   n_{1}(v,\eta) + n_{2}(v,\eta),
\end{array}
\end{equation}

\noindent where $$n_i(v,\eta)= \sum_{u\sim v} 1\{\eta(u)=i\}$$ 
is the number of neighbors of $v$ in state $i$ for configuration $\eta$, for $i\in\{1,2\}.$ Formally, \eqref{rates} means that if node $v$ is in state, say, $0$ at time $t$ then the probability that it will be in state $1$ at time $t+h$, for $h$ small, is $n_{1}(v,\eta) h + o(h)$, where $o(h)$ represents a function such that $\lim_{h\to 0} o(h)/h = 0$. For a review of recent rigorous results for this model on different graphs, we refer the reader to the following references: \cite{EAP,rada,juniormaki,juniorstochastic},  as well as the references cited therein. The primary rigorous results in the existing literature are related to the asymptotic behavior of the proportion of ignorants at the end of the process, in the context of finite networks, or the propagation or non-propagation of rumors in a specific sense, in the case of networks with an infinite number of nodes. The previous quantities we used to evaluate the impact of the information dissemination process within the network. {\color{black}For a deeper discussion of related studies on the spread of online rumors and similar epidemic-like processes we refer the reader to \cite{li,shi,sha,yang,zhu-ding,zhu}, and the references therein.}

\begin{figure*}[h!]
    \centering
    \includegraphics[width=10cm]{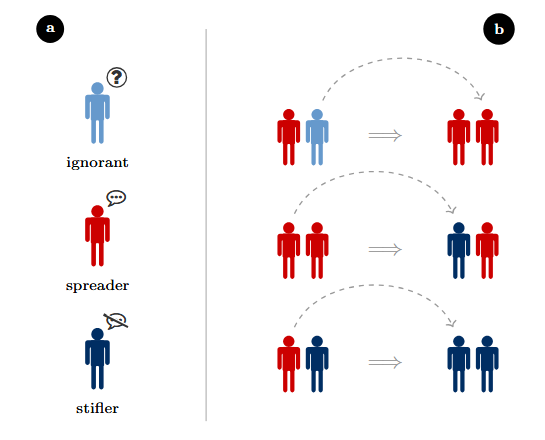}
   \caption{Illustration of all possible transitions between the different classes of individuals in the Maki--Thompson rumor model. (a) The population is subdivided into ignorants, spreaders and stiflers. (b) The different transitions involve interactions between these classes of individuals. When the change is a consequence of the interaction between two individuals, we present the result by a dashed arrow.}\label{FIG:transitionsMT}
\end{figure*}

\subsection{Rigorous results for regular networks} The appropriate approach for dealing with this model depends on the specific network under consideration. It should be noted that networks are mathematically defined as graphs; consequently, some references refer to the subject as an information transmission model on graphs. To the best of our knowledge, rigorous results have been established exclusively for the Maki-Thompson model in the context of ``friendly'' and mainly regular graphs. The model was initially formulated in the context of complete graphs, wherein a homogeneous and completely mixed population was assumed. In this case, the standard approach is to analyze a limiting system of differential equations, which can be derived from the application of well-established techniques for density-dependent Markov chains. For further details, please refer to references
\cite{lebensztayn/machado/rodriguez/2011a,lebensztayn/machado/rodriguez/2011b,rada}. The analysis of differential equations facilitates an examination of the model and enables the attainment of findings pertaining to intriguing generalizations. The approach used to deal with a different graph topology depends on the specific network in question. In \cite{raey}, the authors defined a rumor model on a $d$-dimensional hypercubic lattice and obtained conditions on the contact rates under which the rumor either survives with positive probability or becomes extinct almost surely. Although the model considered by the authors is not the Maki-Thompson model, it may be used to gain insight into the behavior of this class of models on a regular network with an infinite number of nodes. The approach adopted in \cite{raey} involved coupling the original model with an oriented percolation model and applying the results derived from the contact process. See \cite{grimmett} for a deeper discussion on these probabilistic models. Using a different approach,  \cite{juniormaki} studied the Maki-Thompson model and its $k$-stifling variant in infinite Cayley trees. Some of their results were later extended to random trees in  \cite{juniorstochastic}. This case becomes increasingly relevant when we notice that numerous random network models that are better suited to represent a population behave locally as trees. In \cite{juniorstochastic}, researchers employed a comparative approach by examining the model in conjunction with suitably defined branching processes. This approach was extended recently by \cite{puerres} to localize critical thresholds for the rumor model on inhomogeneous trees formed by hubs, regular nodes and leafs. For more expansive or heterogeneous networks, the conventional methodology entails analyzing approximations for the model through computational simulations and mean-field arguments. For illustrative purposes, see references \cite{MNP-PRE2004,moreno-PhysA2007,zanette,zanette02}. 

{\color{black}It is worth pointing out that when the population is represented by a finite network, the previous works have shown that small changes in the network’s topology can produce abrupt differences in how a rumor spreads. If the network is highly clustered, meaning that it contains many short-cuts, a large propagation is expected, usually reaching a positive proportion of the population. Otherwise, the rumor mainly spreads locally, moving through paths of vertices connected only by local edges. In infinite networks, such as tree-like graphs, the cited works have shown that propagation can still occur, even if the proportion of affected individuals remains small. Note that in all cases, adding randomness to the network generally helps reveal different spreading regimes within the same family of network structures.}

\subsection{Stochastic rumors on ring-lattices: What makes it a subject of interest?} 

The ring lattice has often been used as the foundation for small-world network models, which are typically generated by rewiring or adding edges to introduce shortcuts. In a ring lattice with $n$ nodes, the nodes can be arranged in a circle, with each node connected to the $2k$ nearest neighbors. The Watts–Strogatz network is constructed by removing each edge with probability $p$, independent of the other removals, and rewiring it to produce an edge between a pair of nodes selected at random. The Maki-Thompson model on the Watts-Strogatz network was studied in \cite{zanette,zanette02}. It has been shown that this model exhibits a transition between the localization and propagation regimes at a finite value of network randomness. More precisely, the transition occurs between a regime in which the rumor dies in a small neighborhood of the initial spreader and a regime in which it spreads to a positive fraction of the entire population \cite{zanette}. Furthermore, \cite{zanette02} complemented the analysis with a numerical approach and performed a quantitative characterization of the evolution of the two regimes. Another model of a small-world network is the Newman–Watts variant, which is obtained by introducing shortcut edges between pairs of nodes in a similar way to that employed in a Watts–Strogatz network, yet without the removal of edges from the underlying lattice. As proven by \cite{EAP}, this system undergoes a transition between two distinct regimes: localization and propagation. In the localization regime, the final number of stiflers is at most logarithmic in the population size, whereas in the propagation regime, the final number of stiflers grows algebraically with the population size. This transition occurs at a finite value of network parameter $c$. In this model, $c$ represents the mean number of nodes connected to a given node through shortcut edges.

 Note that the results cited emphasize the strong influence of shortcuts on the spreading of information. In this work, we adopt a complementary perspective by focusing on the rumor process within the pure ring lattice, without adding any shortcuts. Our aim is to show that even in this simplified setting, the model can exhibit behavior regarding the proportion of nodes reached by the rumor that is comparable to what is observed in homogeneously mixed populations. To this end, we identify the value of $k$ as a function of $n$ for which this behavior emerges and demonstrate that it scales as $\log n$. 

\subsection{Contribution and organization of the paper} We consider a ring lattice with $n$ nodes, where the nodes are arranged in a circle, and each node is connected to the $2k$ nearest neighbors. We define the rumor process in this network and identify the value of $k$ as a function of $n$ for which the model can exhibit behavior with respect to the proportion of nodes reached by the rumor that is comparable to what is observed in homogeneously mixed populations. Section 2 is devoted to the discussion of the results for two contrasting cases; namely, the complete graph ($k=\lfloor n/2 \rfloor$) and the cycle graph ($k=1$). Section 3 summarizes the numerical results for the general case. In this section, we conclude that the remaining proportion of ignorants is approximately the same as that for the complete graph, provided that $k(n)$ is of the order of $\log n$. 

\subsection{A comment on notation} We consider stochastic processes associated with different graphs. For simplicity, we abuse the notation by omitting the graph dependency, except in cases where the context does not make it clear. In general, we use the notation $\{\eta_t\}_{t\geq 0}$ for a continuous-time Markov chain with state space $\mathcal{S}=\{0,1,2\}^{V}$, where $V$ is the set of nodes of the respective graph, and the process evolves with the rates given by \eqref{rates}. For any $t\geq 0$,  $X(t),Y(t)$, and $Z(t)$ represent the total numbers of ignorants, spreaders, and stiflers at time $t$. Because our work includes a numerical analysis based on Monte Carlo simulations, we used $X(n,k),Y(n,k)$, and $Z(n,k)$ to denote the average number of ignorants, spreaders, and stiflers remaining at the end of the process defined on the $2k$-regular ring lattice with $n$ nodes. The other random variables and averages from the simulations follow similar notation. Because our results are asymptotic, we use the asymptotic notation $f(n)=o(g(n))$ and $f(n)=O(g(n))$, respectively, for functions $f$ and $g$ such that $\lim_{n\to \infty} f(n)/g(n)=0$ and $|f(n)|\leq M |g(n)|$, for all $n\geq n_0$, where $M$ is a positive constant and $n_0$ is a real number. In addition, we use $\lfloor a \rfloor$ ($\lceil a \rceil$) to denote the floor (ceiling) function which maps $a$ to the greatest (least) integer less (greater) than or equal to $a$.

\section{Theoretical results for two contrasting networks} 

We consider a ring lattice with $n$ nodes, where the nodes are arranged in a circle, and each node is connected to the $2k$ nearest neighbors. That is, the set of nodes is $[n]=\{1,\ldots,n\}$, and the set of edges is given by those $(i,j)$ such that $|i-j|\equiv t \mod \left(n-k\right)$, for $0\leq t \leq k$ and $i \neq j$, see Figure \ref{fig:regular}. Note that according to $k$ this family of graphs includes two contrasting examples, which are the complete graph and the cycle graph. In this section, we discuss the behavior of the Maki-Thompson model in these networks. 

\begin{figure*}[h!]
    \centering
    \includegraphics[width=\linewidth]{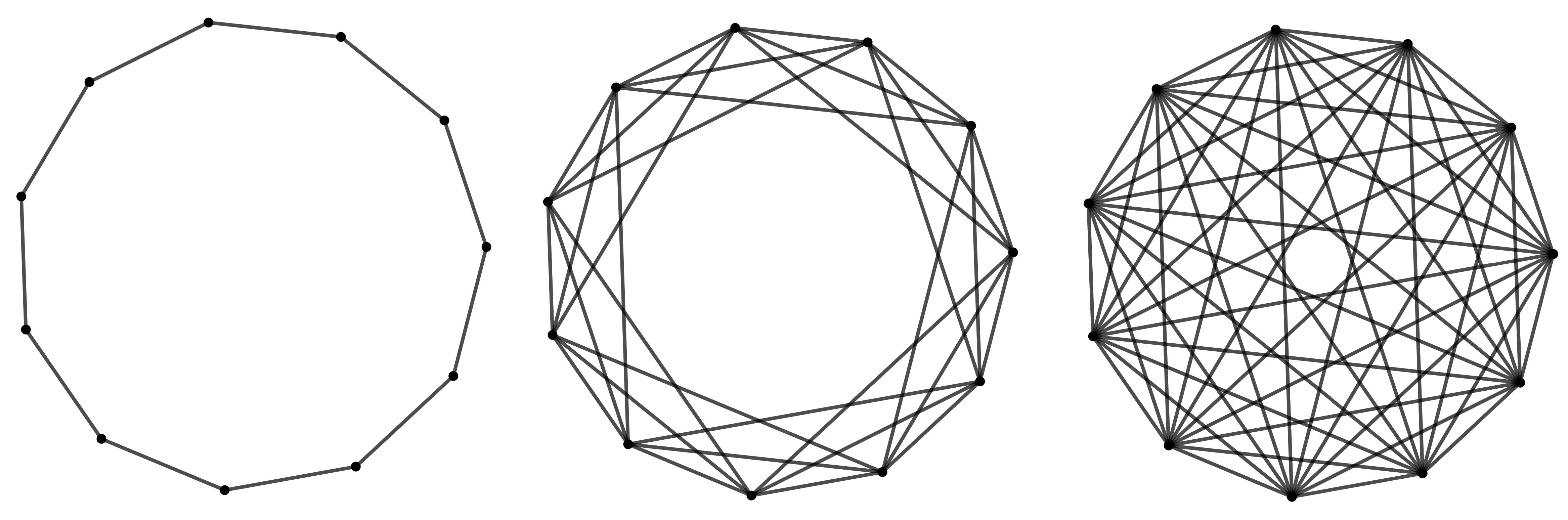}
    \caption{A ring lattice with $n=11$ nodes and each node connected to the $2k$ nearest neighbors. From left to right we have $k=1$ (cycle graph), $k=3$ and $k=5$ (complete graph), respectively.}
    \label{fig:regular}
\end{figure*}

\subsection{Complete graph} The Maki-Thompson rumor model was originally formulated by assuming a homogeneous and mixed population, which is the same to say the  population may be represented by a complete graph. In a complete graph, any pair of nodes is connected via an edge. This is the case $k=\lfloor n/2 \rfloor$ of the ring lattice with $n$ nodes. The random variables of interest are the quantities of ignorants, spreaders, and stiflers at any time $t\geq 0$ which, using the notation of the Markov process $\{\eta_t\}_{t\geq 0}$  with state space $\mathcal{S}=\{0,1,2\}^{[n]}$, $[n]:=\{1,2,\ldots,n\}$, and rates given by \eqref{rates}, can be written as
\begin{equation}\label{eq:variables}
\begin{array}{ccc}
X(t) &=&\displaystyle \sum_{i=1}^{n}\mathbb{I}_{\{\eta_t(i)=0\}},\\[.2cm]
Y(t)&=&\displaystyle \sum_{i=1}^{n}\mathbb{I}_{\{\eta_t(i)=1\}},\\[.2cm]
Z(t)&=&\displaystyle \sum_{i=1}^{n}\mathbb{I}_{\{\eta_t(i)=2\}},
\end{array}
\end{equation}
where $\mathbb{I}_A$ denotes the indicator random variable of event $A$. Note that $X(t)+Y(t)+Z(t)=n$ for any $t\geq 0$. Thus, the Maki-Thompson rumor model on a complete graph may be studied through a continuous-time Markov chain $\{(X(t),Y(t))\}_{t\geq 0}$. The absorption time of the process is the random variable
\begin{equation}\label{eq:taun}
\tau^{(n)}:=\inf\{t\geq 0: Y(t)=0\},
\end{equation}
and counts the time at which there are no spreaders in the population. Thus, the remaining proportion of ignorants is given by $X(\tau^{(n)})/n$, and the remaining proportion of stiflers is $Z(\tau^{(n)})/n=1-X(\tau^{(n)})/n$. The asymptotic behavior of these random quantities is well known and summarized in the following theorems.

\begin{theorem}\label{thm:LLNKn} \cite{sudbury,lebensztayn/machado/rodriguez/2011a,lebensztayn/machado/rodriguez/2011b} Consider the Maki-Thompson model on the complete graph with $n$ nodes, $\{X(t),Y(t)\}_{t\geq 0}$, with $X(0)=n-1$ and $Y(0)=1$. Then
\begin{equation*}
    \lim_{n \to \infty}\frac{X(\tau^{(n)})}{n}=x_\infty \ \ \text{in probability},
\end{equation*}
where 
\begin{equation}\label{eq:xinf}
x_\infty=-\left(\frac{1}{2}\right)W(-2e^{-2})\approx 0.203188,
\end{equation}
where $W$ represents the Lambert $W$ function, which is the inverse of the function $x \mapsto xe^x$.
\end{theorem}

\begin{theorem}\label{thm:TCLKn} \cite{watson,lebensztayn/machado/rodriguez/2011a,lebensztayn/machado/rodriguez/2011b}
Consider the Maki-Thompson model on the complete graph with $n$ nodes, $\{(X(t),Y(t))\}_{t\geq 0}$, with $X(0)=n-1$ and $Y(0)=1$. Then
\begin{equation*}
    \sqrt{n}\left( \frac{X(\tau^{(n)})}{n}-x_\infty\right) \overset{\mathcal{D}}{\rightarrow} \mathcal{N}(0, \sigma^2)\  \text{as}\ n\to \infty,
\end{equation*}
where $\overset{\mathcal{D}}{\rightarrow}$ denotes convergence in distribution, and $\mathcal{N}(0,\sigma^2)$ is the Gaussian distribution with mean zero and variance given by
$$\sigma^2=\frac{x_\infty(1-x_\infty)(1-2x_\infty)}{(1-2x_\infty)^2}.$$
\end{theorem}

Theorem \ref{thm:LLNKn} and Theorem \ref{thm:TCLKn} were initially proved by \cite{sudbury} and \cite{watson}, respectively, appealing to martingale arguments. Recently, \cite{lebensztayn/machado/rodriguez/2011a,lebensztayn/machado/rodriguez/2011b} extended these results to a wide range of rumor models. The main method of \cite{lebensztayn/machado/rodriguez/2011a,lebensztayn/machado/rodriguez/2011b} to prove these theorems is, by means of a random time change, to define a new process with the same transitions as the original, so that they terminate at the same point. This transformation is performed in such a way that the new process is a density dependent Markov chain to which we can apply well-known results of convergence for these Markov processes, see \cite{lebensztayn/machado/rodriguez/2011a,lebensztayn/machado/rodriguez/2011b,rada}. Although \cite{lebensztayn/machado/rodriguez/2011a} also stated a result about the mean number of transitions up to the end of the process, citing arguments of \cite[Theorem 2.5]{kurtz}, we point out that a more general result holds for the Maki-Thompson model on any connected finite network. 

\begin{theorem}\label{thm:timesknew} Consider the Maki-Thompson model on the network $\mathcal{G}$ with $n$ nodes, $\{\eta_{t}\}_{t\geq 0}$. Fix $v\in [n]$, and assume $\eta_0(v)=1$ and $\eta_0(v)=0$ for all $i\in [n]\setminus \{v\}$. Let $M_n$ be the number of transitions that the process makes until absorption, and, for $\mathcal{G}$, let $X(t)$, $Y(t)$, $Z(t)$ be defined as \eqref{eq:variables} and let $\tau^{(n)}$ be defined as \eqref{eq:taun}. Then,
\begin{equation}\label{eq:MR}
M_n=2 Z(\tau_n) -1.
\end{equation}
\end{theorem}

\begin{proof} Consider the Maki-Thompson model on $\mathcal{G}$, denoted by $\{\eta_t\}_{t\geq 0}$. We disregard the time spent in each state and focus only on the transitions, that occur at times $T_1, T_2, \ldots$ Define the discrete-time Markov chain $\{\eta_n\}_{n\geq 0}$ by setting $\eta_n := \eta_{T_n}$, representing the state of the original process immediately after the $n$-th transition. This chain is called the skeleton of the Maki-Thompson model on $\mathcal{G}$. To prove equation \eqref{eq:MR}, we observed that the number of transitions made by the Maki-Thompson model until absorption, denoted by $M_n$, corresponds to the absorption time of its skeleton. Note that at the end of the process, there are $Z(\tau_n)$ stiflers. At each step of the skeleton, there are only two possible outcomes: either a new spreader is added, or a spreader is lost by becoming a stifler. Among the final stiflers, one originates from the initial spreader, node $v$. The remaining $Z(\tau_n) - 1$ stiflers are individuals who were initially ignorant, became spreaders during the process, and later turned into stiflers. The first stifler corresponds to one step of the process, whereas each of the remaining stiflers involves two steps: one for becoming a spreader and the other for becoming a stifler. Therefore, the total number of transitions in the process is $M_n=2(Z(\tau_n)-1)+1$ which proves the result.  
\end{proof}

\begin{corollary}\label{thm:timeskn} \cite{kurtz, lebensztayn/machado/rodriguez/2011a} Consider the Maki-Thompson model on the complete graph with $n$ nodes, $\{(X(t),Y(t))\}_{t\geq 0}$, with $X(0)=n-1$ and $Y(0)=1$. Let $m^{(n)}$ be the mean number of transitions that the process undergoes until absorption. Then, 
\begin{equation*}
    \lim_{n\to \infty} \frac{m^{(n)}}{n}=\tau_\infty,
\end{equation*}
where $\tau_\infty=2(1-x_\infty)$. 
\end{corollary}

\begin{proof}
    This is a direct consequence of Theorem \ref{thm:LLNKn}, Theorem \ref{thm:timeskn} and Lebesgue dominated convergence \cite[Corollary 6.3.2]{resnick}.
\end{proof}

The proof of Corollary \ref{thm:timeskn} suggested by \cite{lebensztayn/machado/rodriguez/2011a} follows the same ideas as the proof of \cite[Theorem 2.5]{kurtz}. Although the latter is a model of interacting random walks on a complete graph, \cite{EP} showed that it is related to the rumor model studied by \cite{lebensztayn/machado/rodriguez/2011a}, which in turn extends the Maki-Thompson rumor model. Our proof of Theorem \ref{thm:timesknew} is intuitive and direct, and generalizes the result stated without proof by \cite{EAP} for the Newman-Watts small-world variant.   

\subsection{The cycle graph} Consider the case $k=1$, which corresponds to the cycle graph, and the Maki-Thompson rumor model is defined on this graph. It is not difficult to intuitively understand that the rumor will be confined to a neighborhood of nodes whose number is independent of $n$. This is because, for $n$ sufficiently large, the number of informed nodes in any direction is bounded from above for a random variable with a geometric distribution with parameter $1/2$. To formalize this idea, we begin by considering the Maki-Thompson model on $\mathbb{Z}$ maintaining the main notation stated in the previous section. That is, consider the Markov process $\{\eta_{t}\}_{t\geq 0}$ with the state space $\mathcal{S}=\{0,1,2\}^{\mathbb{Z}}$ and transition rates given by \eqref{rates}. In addition, we define the absorption time of the process as $\tau^{\infty}:=\inf\{t\geq 0: Y(t)=0\}$, where
$$Y(t)=\sum_{i\in \mathbb{Z}}\mathbb{I}_{\{\eta_t(i)=1\}},$$
for any $t\geq 0$. Consider also the random variables given by:
$$X(t)=\sum_{i\in \mathbb{Z}}\mathbb{I}_{\{\eta_t(i)=0\}} \text{ and }Z(t)=\sum_{i\in \mathbb{Z}}\mathbb{I}_{\{\eta_t(i)=2\}}.$$

\begin{prop}\label{prop:MTZ}
    Consider the Maki-Thompson model on $\mathbb{Z}$, $\{\eta_{t}\}_{t\geq 0}$, with $\eta_0(0)=1$ and $\eta_0(i)=0$ for all $i\in \mathbb{Z}\setminus \{0\}$. Let $R\sim Geom(1/2)$, $L\sim  Geom(1/2)$ and let $Y$ be a discrete random variable with $\mathbb{P}(Y=-1)=\mathbb{P}(Y=1)=1/4=1-P(Y=2)$. Then,
    \begin{equation}\label{eq:MTZ}
        Z(\tau^{\infty})=1 + R\, \mathbb{I}_{\{Y\neq -1\}} + L\, \mathbb{I}_{\{Y\neq 1\}}.
    \end{equation}
Moreover, $E(Z(\tau^{\infty}))=4.$
\end{prop}

\begin{proof}We denote by $\mathbb{Z}^{+}$ ($\mathbb{Z}^{-}$) the set of positive (negative) integers. Let $R:=\sup \{i\in \mathbb{Z}^{+}\cup\{0\}:\eta_{t}(i)\neq 0\text{ for some }t\geq 0\}$ and $L:=\inf \{i\in \mathbb{Z}^{-}\cup\{0\}:\eta_{t}(i)\neq 0\text{ for some }t\geq 0\}$. That is, $R$ denotes the rightmost informed node during the process and $L$ denotes the leftmost informed node during the process. It is not difficult to see that given that $\eta_{t}(1)\neq 0$ for some $t>0$, $R\sim Geom(1/2)$. Similarly, given that $\eta_{t}(-1)\neq 0$ for some $t>0$, $L\sim Geom(1/2)$. Therefore, $Z(\tau^{\infty})$ is equal to 
    $$ 1 + R\, \mathbb{I}_{\displaystyle\left\{\displaystyle\cup_{t\in(0,\infty)}\{\eta_t(1)\neq 0\}\right\}} +  L\, \mathbb{I}_{\displaystyle\left\{\displaystyle\cup_{t\in(0,\infty)}\{\eta_t(-1)\neq 0\}\right\}},$$
    and \eqref{eq:MTZ} is obtained by letting  $Y$ equals to:
    $$\left\{
    \begin{array}{rl}
         1, &\text{ if }\displaystyle\bigcup_{t\in(0,\infty)}\{\eta_t(1)\neq 0\}, \text{ and } \displaystyle\bigcap_{t\in(0,\infty)}\{\eta_t(-1)= 0\}\\[.2cm]
         -1, &\text{ if }\displaystyle\bigcup_{t\in(0,\infty)}\{\eta_t(-1)\neq 0\}, \text{ and } \displaystyle\bigcap_{t\in(0,\infty)}\{\eta_t(1)= 0\}\\[.2cm]
         2,& \text{ if }\displaystyle\bigcup_{t\in(0,\infty)}\{\eta_t(1)\neq 0\}, \text{ and } \displaystyle\bigcup_{t\in(0,\infty)}\{\eta_t(-1)\neq 0\}.
    \end{array}\right.
    $$
Moreover,
    $ E(Z(\tau^{\infty})) =  1 + \mathbb{E}(R\, \mathbb{I}_{\{Y\neq -1\}}) + \mathbb{E}(L\, \mathbb{I}_{\{Y\neq 1\}})= 1+ (3/4)\,\mathbb{E}(R) + (3/4)\,\mathbb{E}(L)= 4.$
\end{proof}

Proposition \ref{prop:MTZ} is of particular interest when one considers the behavior of the Maki-Thompson rumor model in the cycle graph with $n$ nodes, which is asymptotically the same as in $\mathbb{Z}$.

\begin{theorem}\label{thm:MTk1}
    Consider the Maki-Thompson model in a cycle graph with $n$ nodes, $\{\eta_{t}\}_{t\geq 0}$, with $\eta_0(1)=1$ and $\eta_0(i)=0$ for all $i\in [n]\setminus \{1\}$. Let $R\sim Geom(1/2)$, $L\sim  Geom(1/2)$ and $Y$ be discrete random variables with $\mathbb{P}(Y=-1)=\mathbb{P}(Y=1)=1/4=1-P(Y=2)$. Then,
    \begin{equation}\label{eq:limTzRK}
        \lim_{n\to \infty} Z(\tau^{(n))})=1 + R\, \mathbb{I}_{\{Y=1\}} + L\, \mathbb{I}_{\{Y=-1\}} \ \ \text{almost surely}.
    \end{equation}
Moreover, $\displaystyle\lim_{n\to \infty} E(Z(\tau^{(n))}))= 4.$
\end{theorem} 

\begin{proof} The proof is by coupling, for any $n$, of the Maki-Thompson model in the cycle graph with $n$ nodes, $\{\eta_{t}^{(n)}\}_{t\geq 0}$, and the Maki-Thompson model on $\mathbb{Z}$. Let $R_n:=\max \{i\in \{2,3,\ldots \lceil n/2 \rceil\}:\eta_{t}^{(n)}(i)\neq 0\text{ for some }t\geq 0\}$ and $L_n:=\min \{i\in \{-\lfloor n/2 \rfloor, \ldots, -2,-1\}:\eta_{t}^{n}(i)\neq 0\text{ for some }t\geq 0\}$. Then
        \begin{equation}
            R_n =\left\{
            \begin{array}{cl}
               R,  & \text{ if }R < \lceil n/2 \rceil \\
               \lceil n/2 \rceil, & \text{ if }R \geq \lceil n/2 \rceil.
            \end{array}\right.
        \end{equation}
        and
            \begin{equation}
            L_n =\left\{
            \begin{array}{cl}
               L,  & \text{ if }L > - \lfloor n/2 \rfloor \\
               - \lfloor n/2 \rfloor, & \text{ if }L \leq -\lfloor n/2 \rfloor.
            \end{array}\right.
        \end{equation}
        
        Therefore, the result is a consequence of Proposition \ref{prop:MTZ} and the fact that, by construction, $R_n \rightarrow R$ and $L_n \rightarrow L$ almost surely.
\end{proof}

Now, we can compare with the result for the complete graph, that is, with Theorem \ref{thm:LLNKn}.

\begin{corollary}\label{cor:xinfRK}
      Consider the Maki-Thompson model in a cycle graph with $n$ nodes, $\{\eta_{t}\}_{t\geq 0}$, with $\eta_0(1)=1$ and $\eta_0(i)=0$ for all $i\in [n]\setminus \{1\}$. Let $X(t)$ be defined as \eqref{eq:variables} and let $\tau^{(n)}$ be defined as \eqref{eq:taun}. Then,
\begin{equation*}
    \lim_{n \to \infty}\frac{X(\tau^{(n)})}{n}=1 \ \ \text{almost surely.}
\end{equation*}
\end{corollary}

\begin{proof}
    By \eqref{eq:limTzRK} and the fact that $X(\tau^{(n)})=n-Z(\tau^{(n)})$.
\end{proof}

{\color{black}
Modelling rumor spreading on networks as an interacting particle system is recent. In this section, we explore the coupling technique between the original stochastic process on the finite cycle graph and the same process defined on the line. This is a constructive approach that can be adapted to the study of other contagion processes on ring-like lattices.}

\section{Numerical results for the general case}

From the results summarized in the previous section, we conclude that if $X(t)$ and $\tau^{(n)}$ are defined for the $2k$-regular ring lattice with $n$ nodes by \ref{eq:variables} and \eqref{eq:taun}, respectively, then
$$\lim_{n\to \infty}\frac{X(\tau^{(n)})}{n}\,
\left\{
\begin{array}{clll}
= & 1& \text{ a.s. }, & \text{ if }k=1;\\[.2cm]
 \approx & 0.203188 & \text{ a.s. }, & \text{ if }k=\lfloor n/2\rfloor.
\end{array}
\right.
$$

A natural question that arises is whether there is a critical $k$ as a function of $n$ above which the remaining proportion of ignorants is $\approx 0.203$. The purpose of this section is to answer this question. 

\subsection{Simulations methodology}
In this section we present numerical results obtained by simulating the Maki-Thompson model on a $2k$-regular ring lattice with $n$ nodes. We are interested in measuring the remaining proportion of ignorants, as $k$ varies with $n$. Thus, to emphasize the dependence of the parameters, we denote the average final number of ignorants as $X(n,k)$. Analogously, denoted by $Z(n,k)$ the average final number of stiflers and let $x(n,k):=X(n,k)/n$ and $z(n,k):=Z(n,k)/n$ are the remaining average proportion of ignorants and stiflers, respectively. Before presenting the numerical results, let us provide a brief commentary on the methodology employed in the simulations. In order to implement the Maki-Thompson process on the ring lattice we initially extract randomly a node, say $v$, which will play as the original spreader, while the remaining $n-1$ nodes are ignorant; as a consequence we create the sets of spreaders, ignorants and stiflers, respectively, as $\mathcal{Y}(0) = \{v\}$, $\mathcal{X}(0) = [n]\setminus\{v\}$ and $\mathcal{Z}(0)=\emptyset$. Then, in the next time steps $t = 1, 2, \ldots$, we sequentially repeat the following rules. Consider the set of spreaders $\mathcal{Y}(t-1)$, and randomly extract a node $u \in \mathcal{Y}(t-1)$. From the set of neighbors, $\mathcal{N}(u)$ of $u$, randomly extracts node $\omega$. If the neighbor $\omega$ is ignorant, then $\omega$ turns into a spreader and, accordingly, $\mathcal{Y}(t + 1) = \mathcal{Y}(t) \cup \{\omega\}$ while $\mathcal{X}(t + 1) = \mathcal{X}(t) \setminus \{\omega\}$, otherwise, if the neighbor $\omega$ is either spreader or stifler, then $u$ turns into a stifler and, accordingly, $\mathcal{Y}(t + 1) = \mathcal{Y}(t) \setminus \{u\}$ while $\mathcal{Z}(t + 1) = \mathcal{Z}(t) \cup \{u\}$. This procedure is repeated as long as there are spreaders among individuals, that is, as long as the cardinality of the set of spreaders is greater than zero. The final time is designated as $t(n,k)$, signifying that the number of spreaders is zero at all times greater than $t(n,k)$. Conversely, for all $t < t(n,k)$, the number of ignorants represented by $|\mathcal{X}(t)|$, decreases over time. Upon completion of the simulation, the final value of the number of ignorants, denoted by $|\mathcal{X}(t(n,k))|$, was recorded. Subsequently, the initial setting was reinstated and the simulation was repeated $M$ times. For each repetition $i\in\{1,2,\ldots,M\}$ an estimate of $|\mathcal{X}(t(n,k))|^{(i)}$ is obtained, and the empirical average, denoted by $X(n,k)$, is calculated. That is $X(n,k):=\sum_{i=1}^M |\mathcal{X}(t(n,k))|^{(i)}$. we are interested  in the empirical average of the remaining proportion of ignorants denoted by $x(n,k):=X(n,k)/n$. In addition, let $m(n,k):= \sum_{i=1}^M t(n,k)^{(i)}$ be the average number of transitions up to the end of the process. Note that we should expect $x(n,\lfloor n/2 \rfloor)\approx x_\infty$ and $m(n,\lfloor n/2 \rfloor) \approx n\tau_{\infty}$, where $x_\infty$ and $\tau_{\infty}$ are given by Theorem \ref{thm:LLNKn} and Corollary \ref{thm:timeskn}, respectively. To select the value of $M$, we fix $n=1000$ and obtain a graph of the final proportion of ignorant individuals with a varying number of simulations. We then obtained the average of these simulations and present the results in Figure \ref{fig:Var_Repeticiones}. Because this proportion is not dependent on the number of simulations, we set $M=1000$ to reduce the computational time of the simulations.

\begin{figure*}[htpb]
    \centering
    \subfigure[Average remaining proportion of ignorants.]{\includegraphics[width=0.45\linewidth]{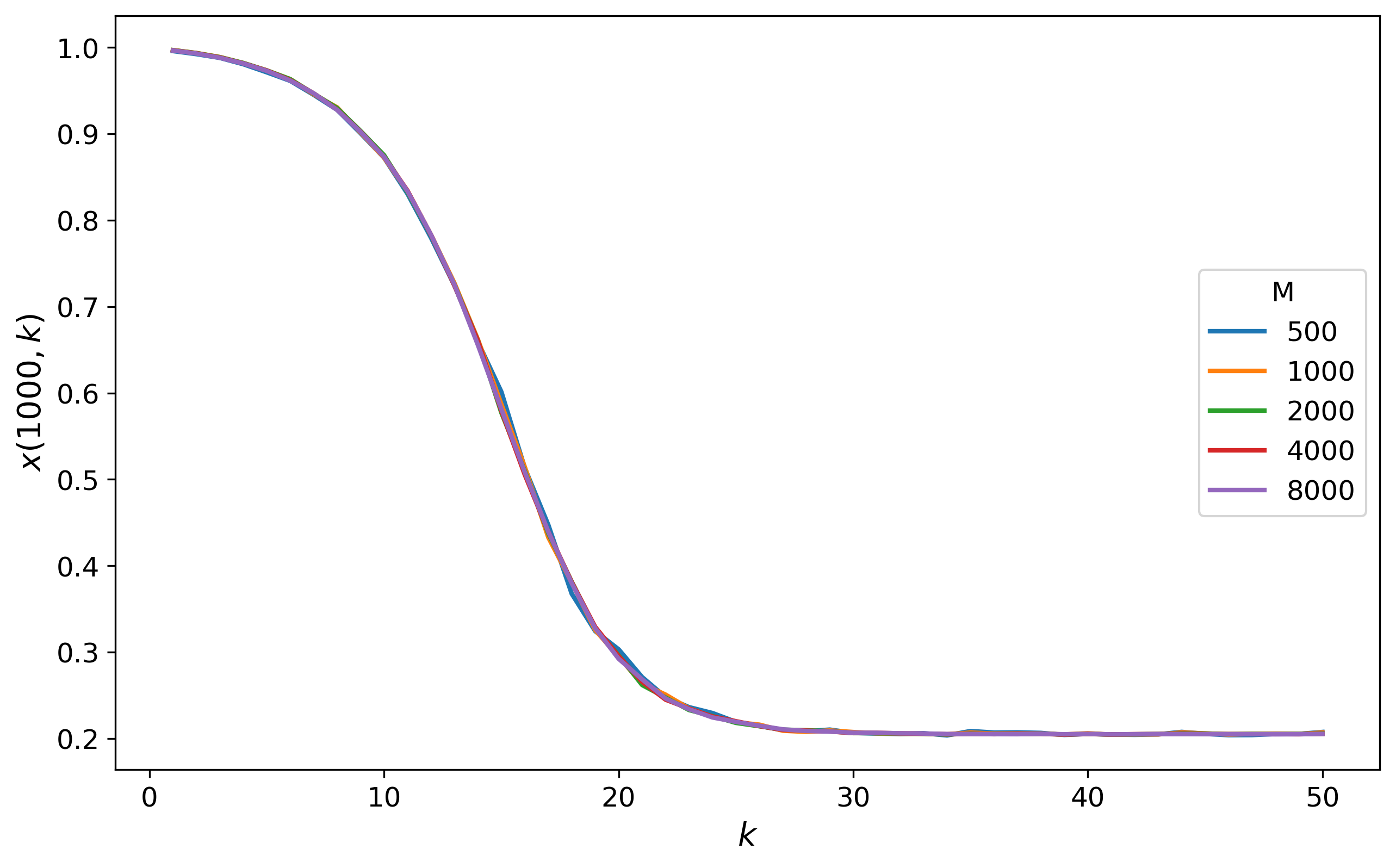}}
    \subfigure[Average number of transitions up to the end of the process.]{\includegraphics[width=0.45\linewidth]{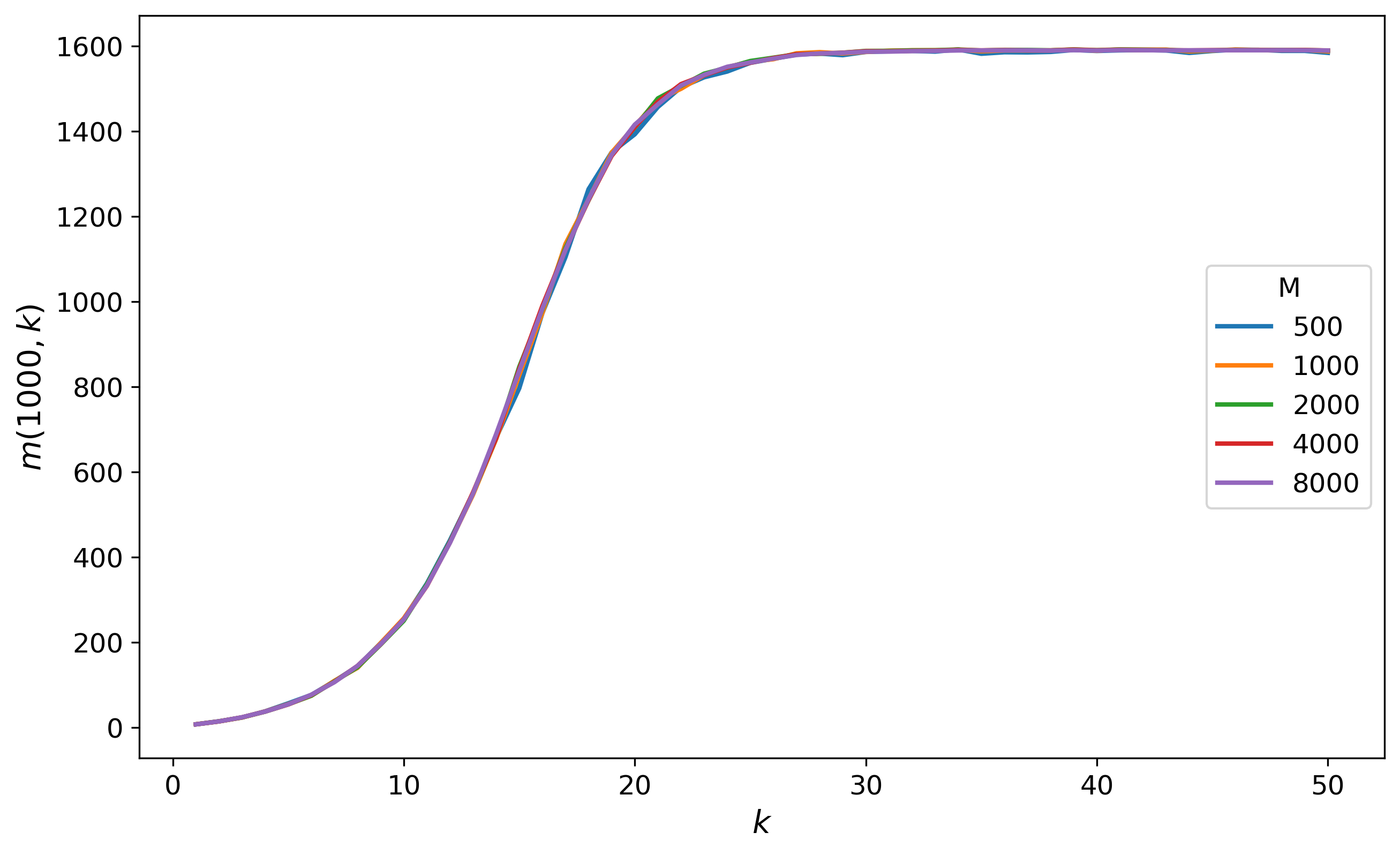}}
    \caption{Behavior of the average limiting quantities of the Maki-Thompson rumor model in the ring lattice according to $k$ for $n=1000$ and varying the number of simulations.}
    \label{fig:Var_Repeticiones}
\end{figure*}

\subsection{The size of the rumor and the duration of the process according to $k=k(n)$} 

To analyze the behavior of the quantities of interest, we performed a series of simulations in which the number of nodes varied between 500 and 32,000. The results of these simulations, which were repeated 1000 times, are shown in Figure \ref{fig:Var_Nodos}.

\begin{figure*}[htpb]
    \centering
    \subfigure[Average remaining proportion of ignorants.]{\includegraphics[width=0.45\linewidth]{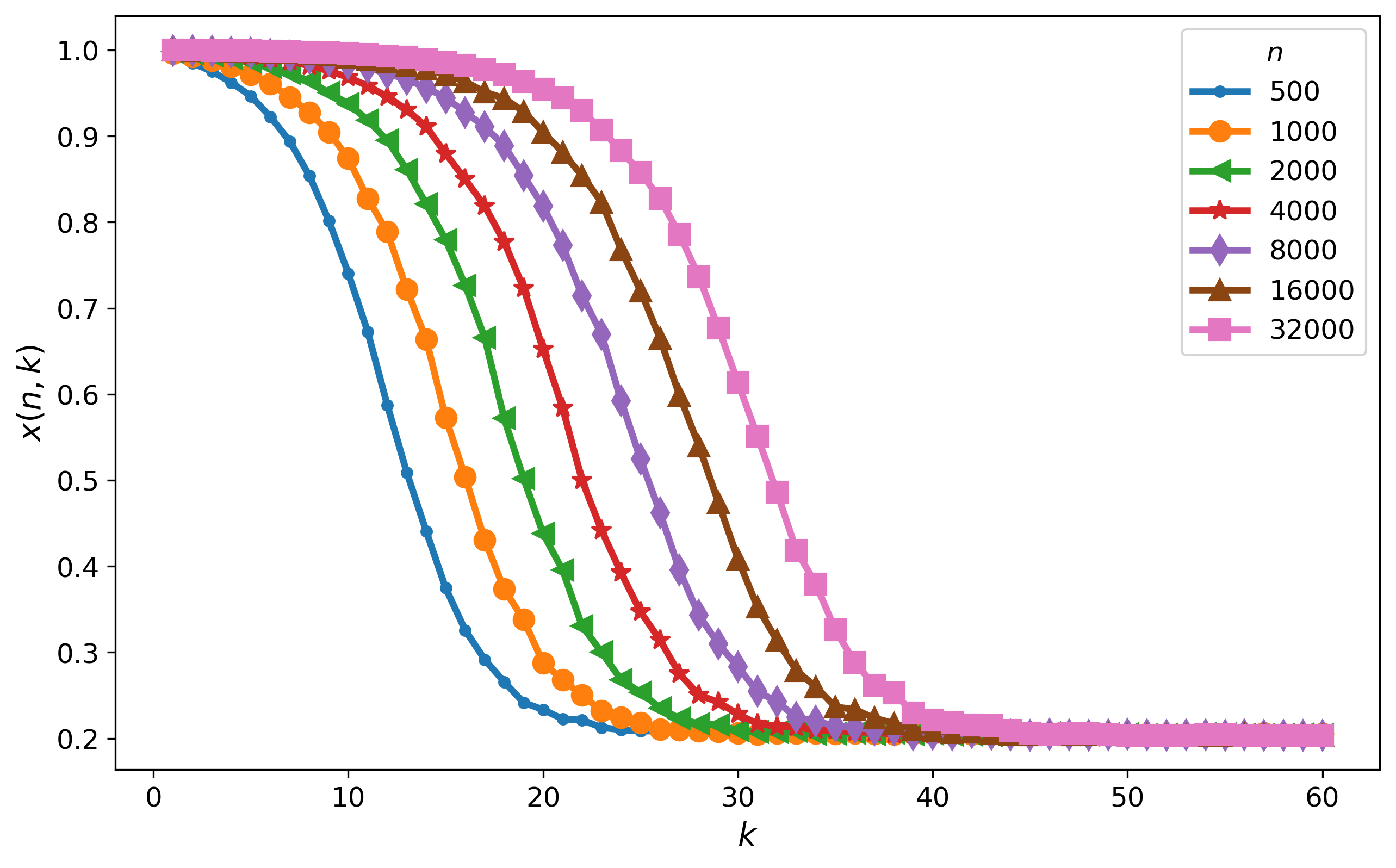}}
    \subfigure[Average number of transitions up to the end of the process.]{\includegraphics[width=0.46\linewidth]{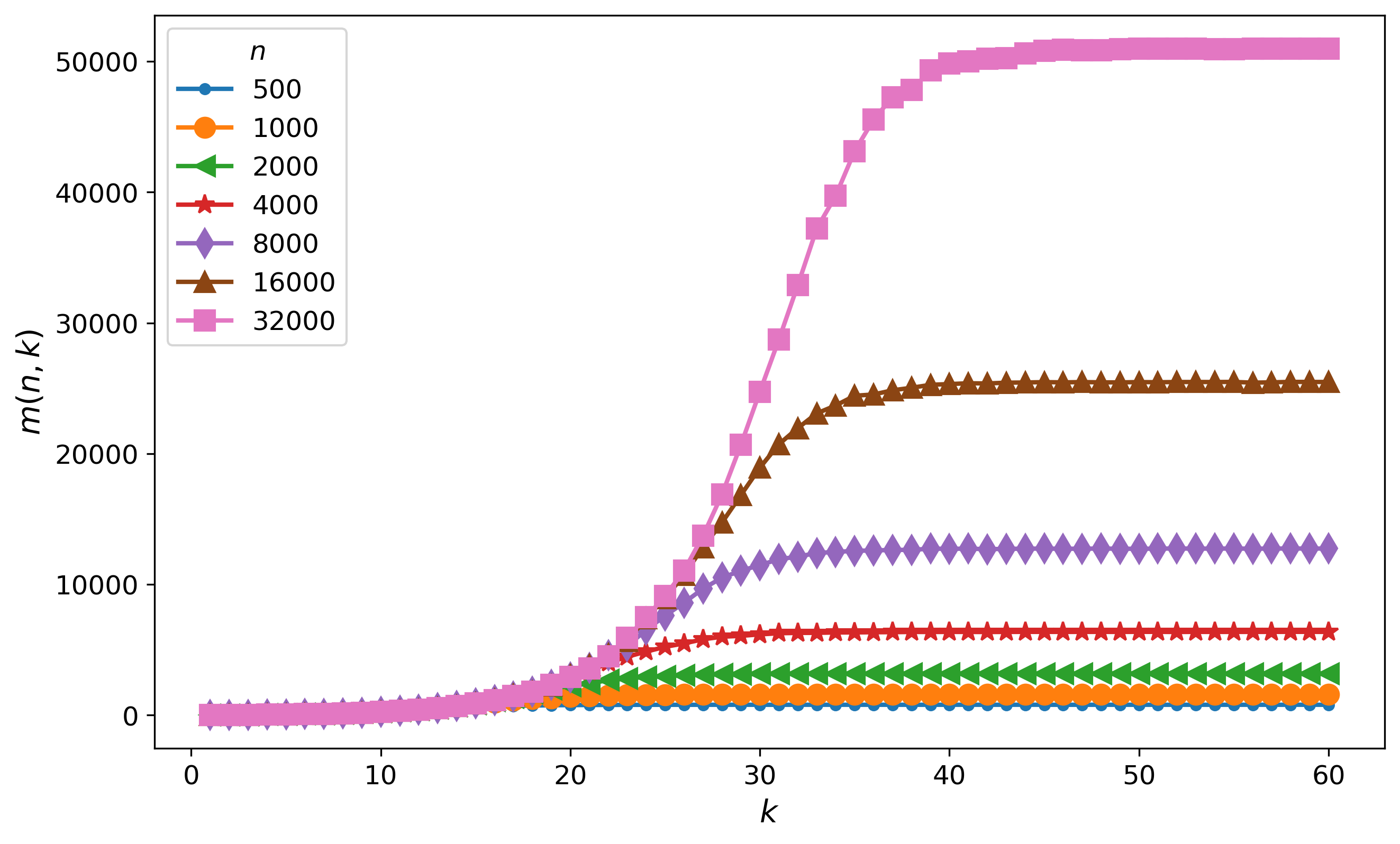}}
       \caption{Behavior of the average limiting quantities of the Maki-Thompson rumor model in the ring lattice according to $k$ for $1000$ simulations and varying the number of nodes. The average proportion of ignorants $x(n,k)$, see (a), and the average number of transitions up to the end of the process $m(n,k)$, see (b), are shown versus $k$ for different values of $n$. Notice that curves for $x(n,k)$ pertaining to different values of $n$ collapse after a certain value of $k$, whereas a similar collapse is observed for all the curves of $m(n,k)$ before a certain value of $k$.}
    \label{fig:Var_Nodos}
\end{figure*}

Furthermore, as illustrated in Figure \ref{fig:Var_Nodos}, the average final proportion of ignorants tends to approach the limiting value for complete graphs, which is $0.203187$, as the value of $k$ increases for each $n$. To substantiate the aforementioned assertion, Table \ref{tab:PropIgn} presents the average proportion derived from the simulation of $1000$ repetitions and varying values of $k$ for each $n$. Thus, it can be concluded that in most cases, the error is accurate to three digits and within a tolerance of $10^{-3}$. To improve this accuracy for a fixed $n$, it is necessary to increase $k$. However, it should be noted that the behavior is asymptotic for the complete graph, which means that while increasing $k$ improves the accuracy, it also causes the computational time to increase exponentially. Therefore, the current accuracy was sufficient.

\begin{table*}[htpb]
\centering
\begin{tabular}{c|llllllll}
$k \backslash n$ & 500      & 1000 & 2000 & 4000 & 8000 & 16000 & 32000 \\\hline
25 & 0.208348 & 0.218133 & 0.253501 & 0.347087 & 0.524593 & 0.720238 & 0.857631 \\
30 & 0.206684 & 0.206166 & 0.209263 & 0.228014 & 0.283021 & 0.408521 & 0.613475 \\
35 & 0.205532 & 0.205462 & 0.205256 & 0.209110 & 0.214324 & 0.236641 & 0.325922 \\
40 & 0.205632 & 0.205873 & 0.205219 & 0.204223 & 0.205050 & 0.208674 & 0.221189 \\
45 & 0.205602 & 0.204440 & 0.204785 & 0.203887 & 0.203780 & 0.204028 & 0.206267 \\
50 & 0.206096 & 0.204476 & 0.204559 & 0.203550 & 0.204891 & 0.203952 & 0.203648 \\
%\cellcolor{yellow}53 & \cellcolor{yellow}{0.206246} & \cellcolor{yellow}{0.204663} &\cellcolor{yellow}{0.203925} & \cellcolor{yellow}{0.203834} & \cellcolor{yellow}{0.203928} & \cellcolor{yellow}{0.203731} & \cellcolor{yellow}{0.203768} \\
55 & 0.207360 & 0.204694 & 0.204927 & 0.204141 & 0.203557 & 0.203543 & 0.204269 \\
60 & 0.207482 & 0.205305 & 0.204097 & 0.204256 & 0.203630 & 0.203714 & 0.203667 \\
%$k \backslash n$ & 500      & 1000 & 2000 & 4000 & 8000 \\\hline
%25      & 0.209840 & 0.215978 & 0.245943 & 0.350333 & 0.522487 \\
%30     & 0.205866 & 0.205709 & 0.211446 & 0.225536 & 0.282944 \\
%35      & 0.207034 & 0.205375 & 0.204768 & 0.209348 & 0.218601 \\
%40      & 0.205204 & 0.205370 & 0.205220 & 0.204831 & 0.205323\\
%45      & 0.206248 & 0.205052 & 0.204319 & 0.204539 & 0.203878\\
%50     & 0.206600 & 0.204949 & 0.204411 & 0.204188 & 0.203905\\
%55     & 0.205382 & 0.204808 & 0.204287 & 0.203926 & 0.204014\\
%60     & 0.205926 & 0.203417 & 0.204214 & 0.204324 & 0.204239    
\end{tabular}
\caption{Average remaining proportion of ignorants $x(n,k)$.}
    \label{tab:PropIgn}
\end{table*}

Note that by fixing $k=1$ and varying the number of nodes, it is numerically verified in Table \ref{tab:Ring} that the final average number of stiflers in the cycle graph is approximately $4$, and the average remaining proportion of ignorants is approximately $1$, verifying the results of Theorem \ref{thm:MTk1} and Corollary \ref{cor:xinfRK}, respectively.

\begin{table*} [htpb]
    \centering
    \begin{tabular}{clllllll}
    \hline
       $n$      & 500 & 1000 & 2000 & 4000 & 8000 & 16000 & 32000 \\ \hline
       $Z(n,1)$ & 3.989 & 4.065 & 4.134 & 3.984 & 3.992 & 3.856 & 4.128 \\ \hline
          $x(n,1)$ & 0.992022 & 0.995935 & 0.997933 & 0.999004 & 0.999501 & 0.999759 & 0.999871 \\ \hline
    
%       $n$   & 500 & 1000 & 3000 & 5000 & 8000 & 10000 & 20000 \\ \hline
%       $X(n,1)$   & 0.991895 & 0.996075& 0.998681& 0.999194& 0.999497& 0.999602& 0.999794\\     \hline  
%       $Z(n,1)$   & 4.0525 & 3.9250 & 3.9570 & 4.0300 & 4.0240 & 3.9800 & 4.1200  \\\hline
    \end{tabular}
    \caption{Average remaining number of stiflers and average remaining proportion of ignorants for $k=1$.}
    \label{tab:Ring}
\end{table*}

It should be noted that Figure \ref{fig:Var_Nodos} also shows that the curves for $x(n,k)$ pertaining to different values of $n$ collapse after a certain value of $k$, while a similar collapse is evidenced for all the curves of $m(n,k)$ before a certain value of $k$. This suggests the existence of a ``critical'' value of $k$, as a function of $n$, after which we could expect a behavior similar to that of the model on the complete graph. We explore this behavior in Figure \ref{fig:Time_completegraphs}. Note that if we observe $m(n,k)$ at the logarithmic scale and focus on doubling the number of nodes, then we obtain almost parallel lines with equal spacing, suggesting the same behavior of the complete graph. Indeed, if we consider Corollary \ref{cor:xinfRK}, for the complete graph, $m^{(n)} \approx 2\,n\,(1-x_\infty).$ Then, for $n$ sufficiently large, we obtain $\log m^{(2n)}-\log m^{(n)}\approx \log 2$, constant, which is observed from simulations for $\log m(2n,k)- \log m(n,k)$ for all values of $k$ from a given point (see Figure \ref{fig:Time_completegraphs} (b)).

\begin{figure*}[h!]
    \centering
    \subfigure[]{\includegraphics[width=0.457\linewidth]{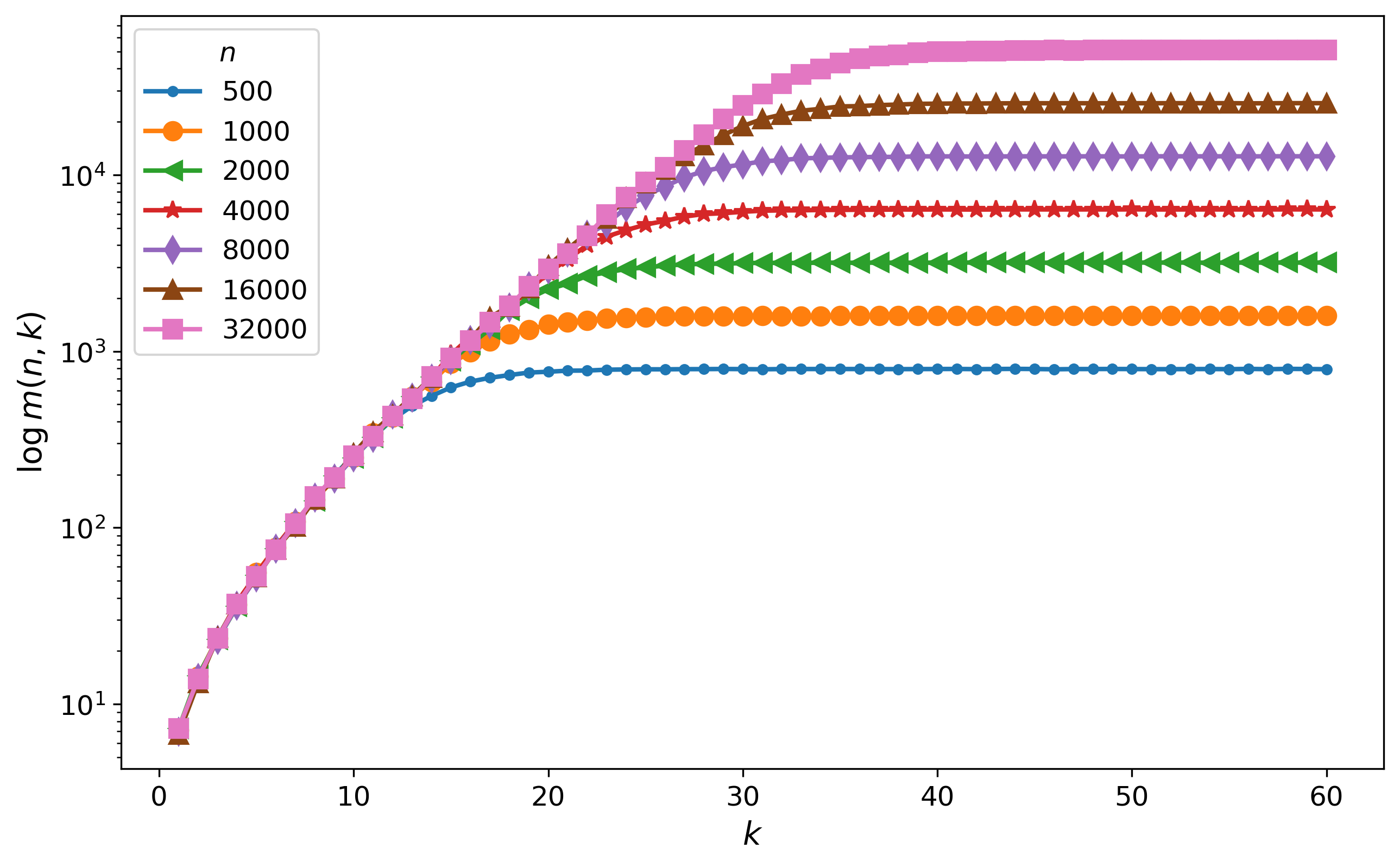}}
    \subfigure[]{\includegraphics[width=0.525\linewidth]{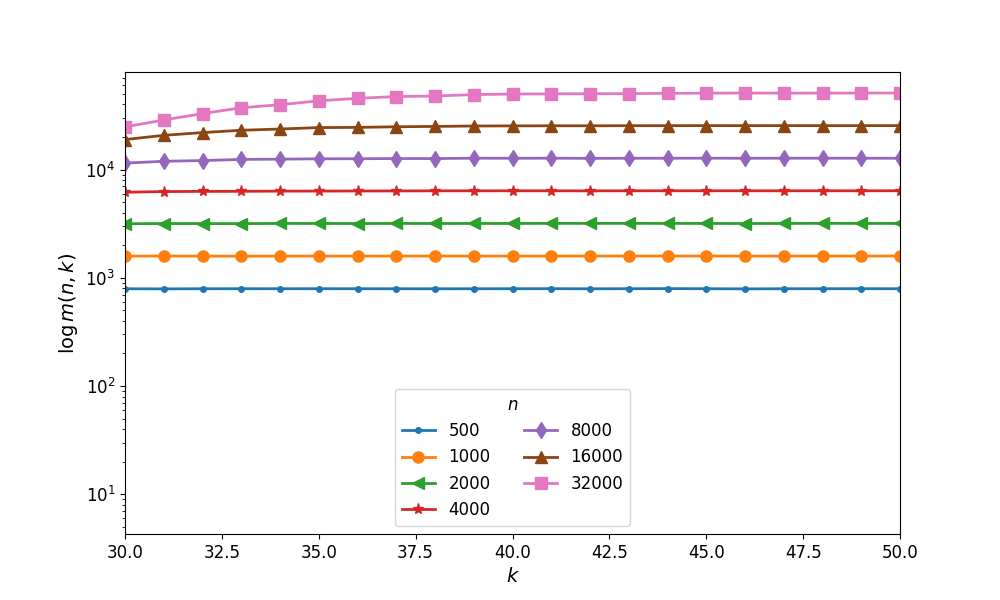}}
    \caption{(a) The average number of transitions up to the end of the process $m(n,k)$, is shown versus $k$ for different values of $n$ at logarithmic scale. (b) From a certain value of $k$, the curves for $m(n,k)$ shown at logarithmic scale and doubling the number of nodes, result in parallel lines with equal spacing suggesting the same behavior of the complete graph (see Corollary \ref{cor:xinfRK}). } 
%    The behavior of the mean number of transitions in the Figure  \ref{fig:Var_Nodos}(b) until the end of the process on the $k$ regular ring graph up $k_\infty$ is the same to the complete graphs  in the Theorem \ref{thm:timeskn}. On a logarithmic time scale, doubling the number of nodes results in parallel lines with equal spacing. The zoom in (b) shows the asymptotic approach. }
    \label{fig:Time_completegraphs}
\end{figure*}

\subsection{Logistic approximation for the final proportion of ignorants}

Note that $x_\infty \leq x(n,k) \leq 1$, where $x_\infty\approx 0.203188$ is the asymptotic remaining proportion of ignorants for the Maki-Thompson model on the complete graph (see \eqref{eq:xinf}). Furthermore, from Figure \ref{fig:Var_Repeticiones} and Figure \ref{fig:Var_Nodos}, it is possible to relate the obtained results to a logistic function in the sense that we can assume $x(n,k)\approx \tilde{x}(n,k)$, for $k\in \{1,2,\ldots\}$, where
\begin{equation}
   \tilde{x}(n,k)=\frac{z_\infty}{1+e^{-\alpha(k-\beta)}}+x_\infty,
    \label{eq:Logistic}
\end{equation}
is defined for $k\in(0,\infty)$, $z_\infty=1-x_\infty\approx 0.796812$ is the amplitude and $x_\infty=0.203188$ is the vertical shift of the function. On the other hand, $\alpha:=\alpha(n)$ and $\beta:=\beta(n)$ are parameters that vary with the number of nodes, with $\alpha<0$ because it is a logistic function reflected with respect to the $y$ axis, and $\beta$ refers to the magnitude of the horizontal shift of the function. Figure \ref{fig:fit} shows the curve $\tilde x(n,k)$, which is a fitting plot obtained for some values of $n$ in blue, and the values from $x(n,k)$, which are the data obtained from the simulation. In addition, the values of $\alpha$ and $\beta$ obtained for each fit are included.  

\begin{figure*}[h!]
    \centering
    \subfigure[$\alpha = -0.370505$ y $\beta = 14.609496$.]{\includegraphics[width=0.45\linewidth]{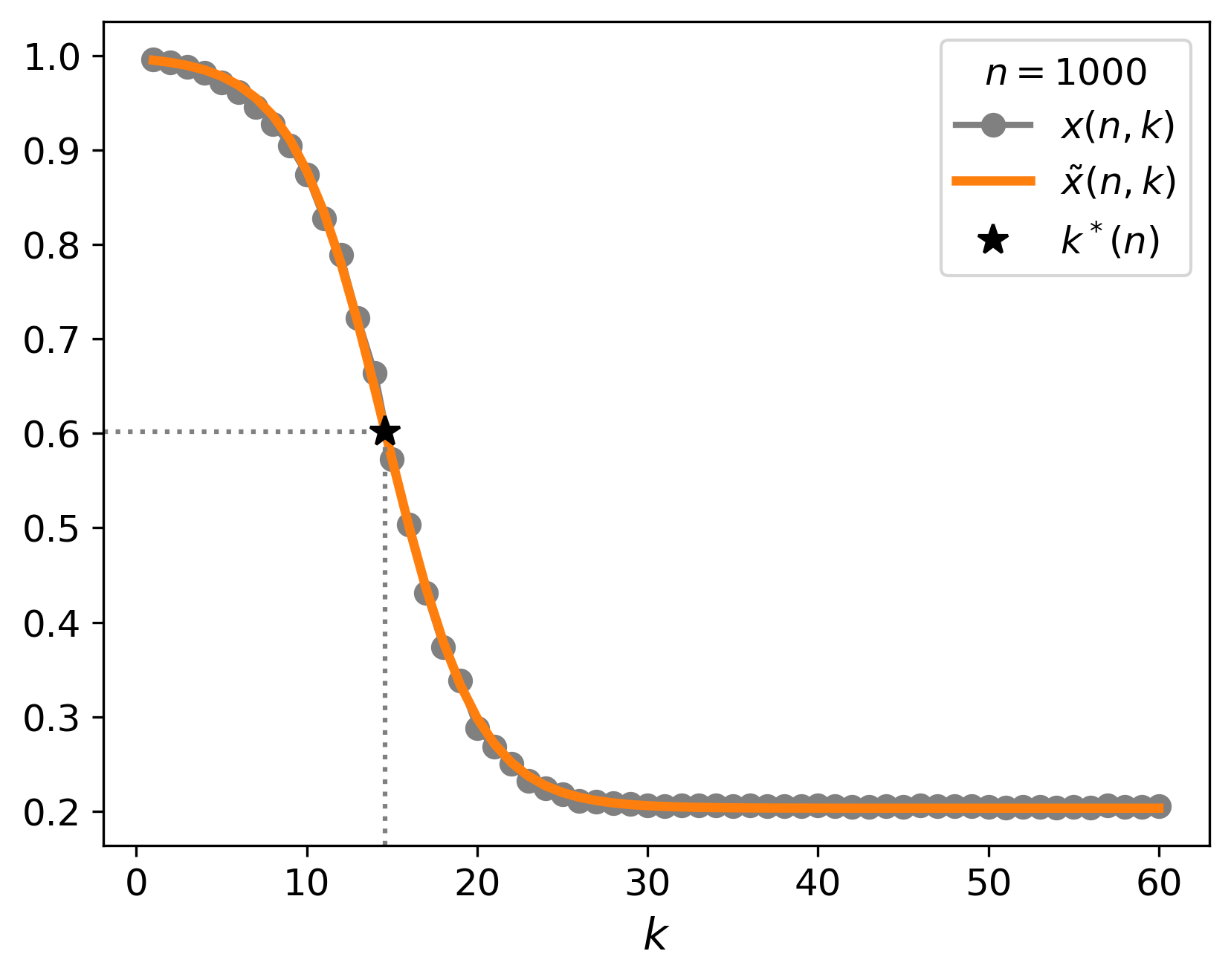}}
    \subfigure[$\alpha = -0.353508$ and $\beta = 17.599477$.]{\includegraphics[width=0.45\linewidth]{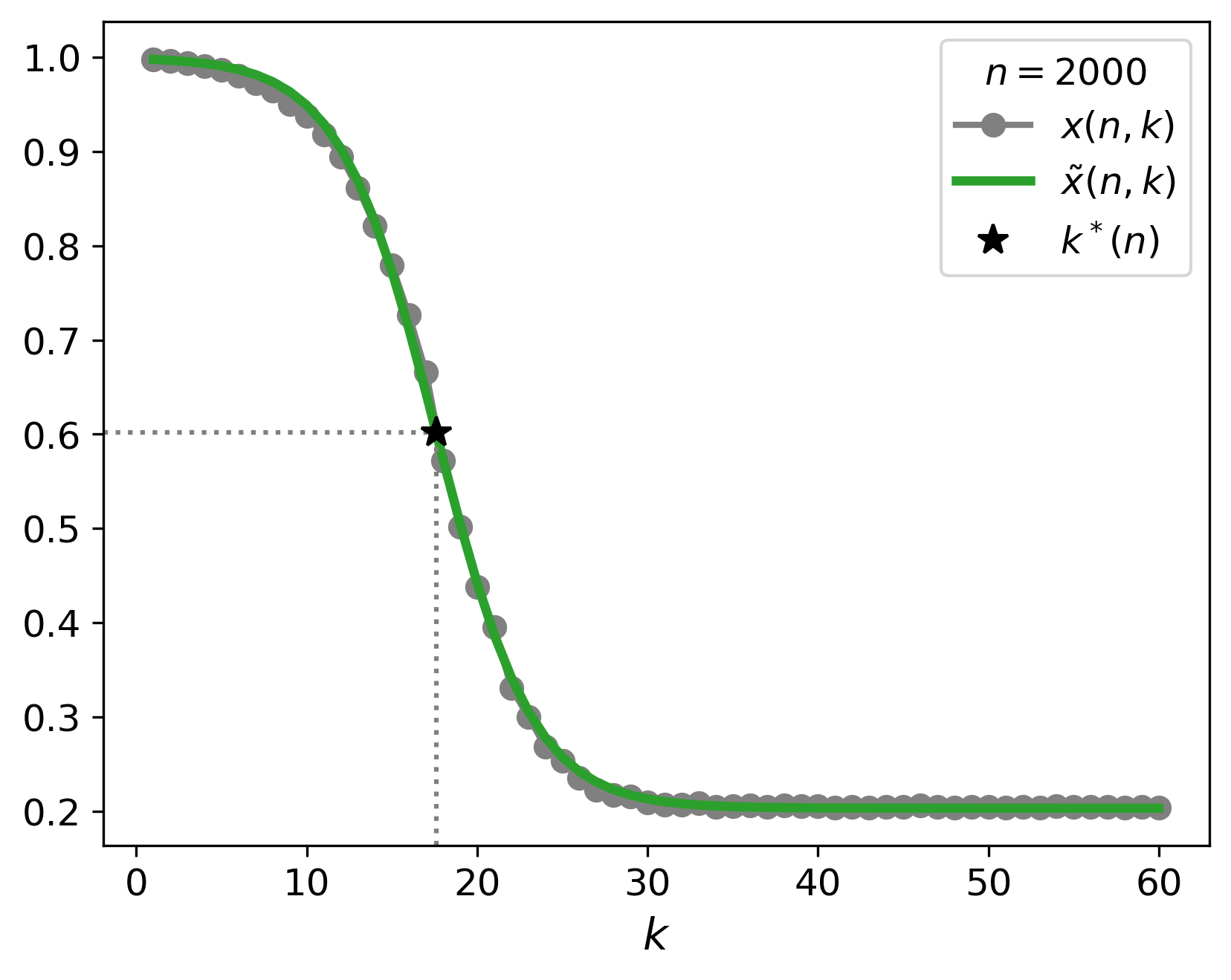}}
    \subfigure[$\alpha = -0.337263$ and $\beta = 20.575646$.]{\includegraphics[width=0.45\linewidth]{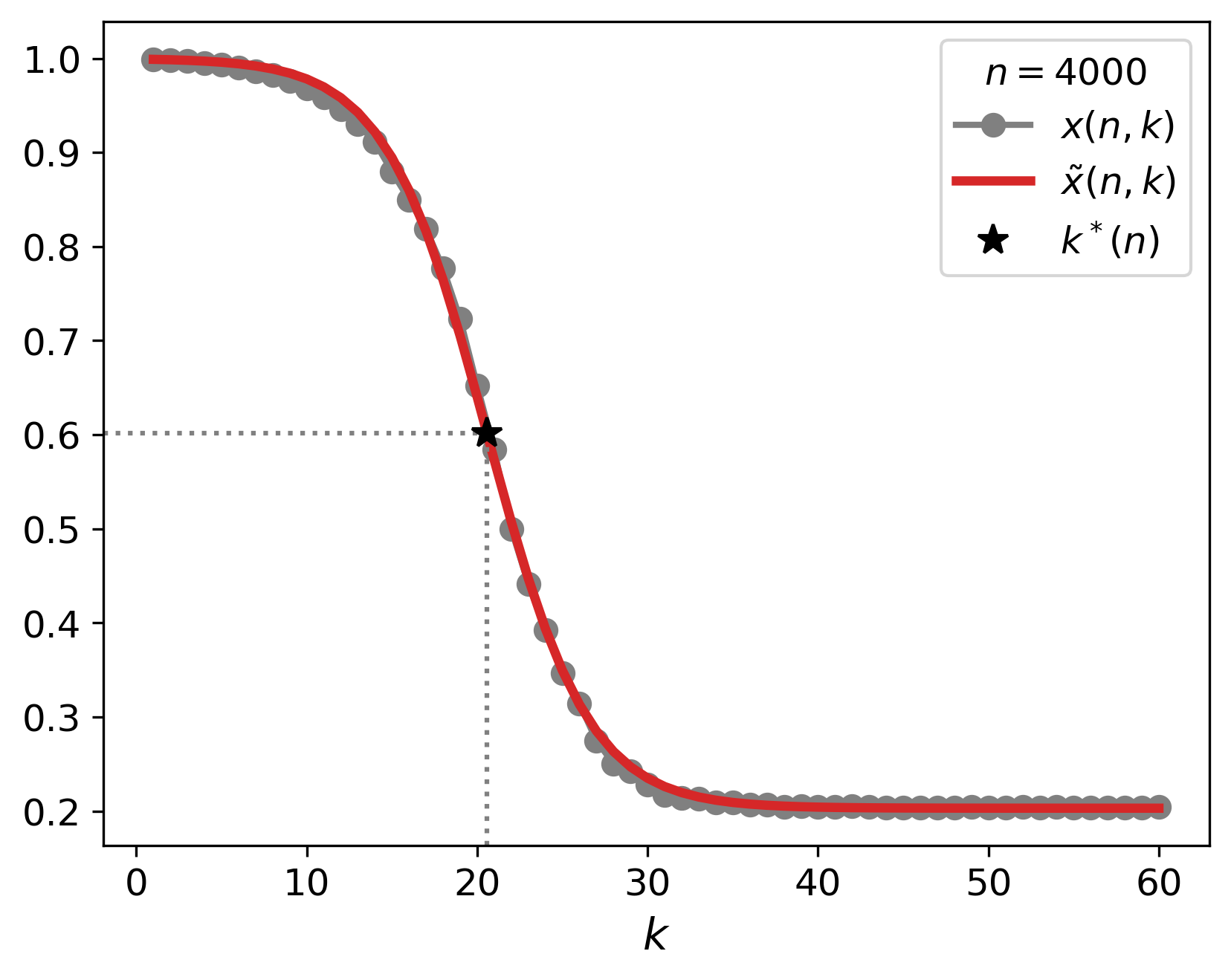}}
    \subfigure[$\alpha = -0.332851$ and $\beta = 23.688934$.]{\includegraphics[width=0.45\linewidth]{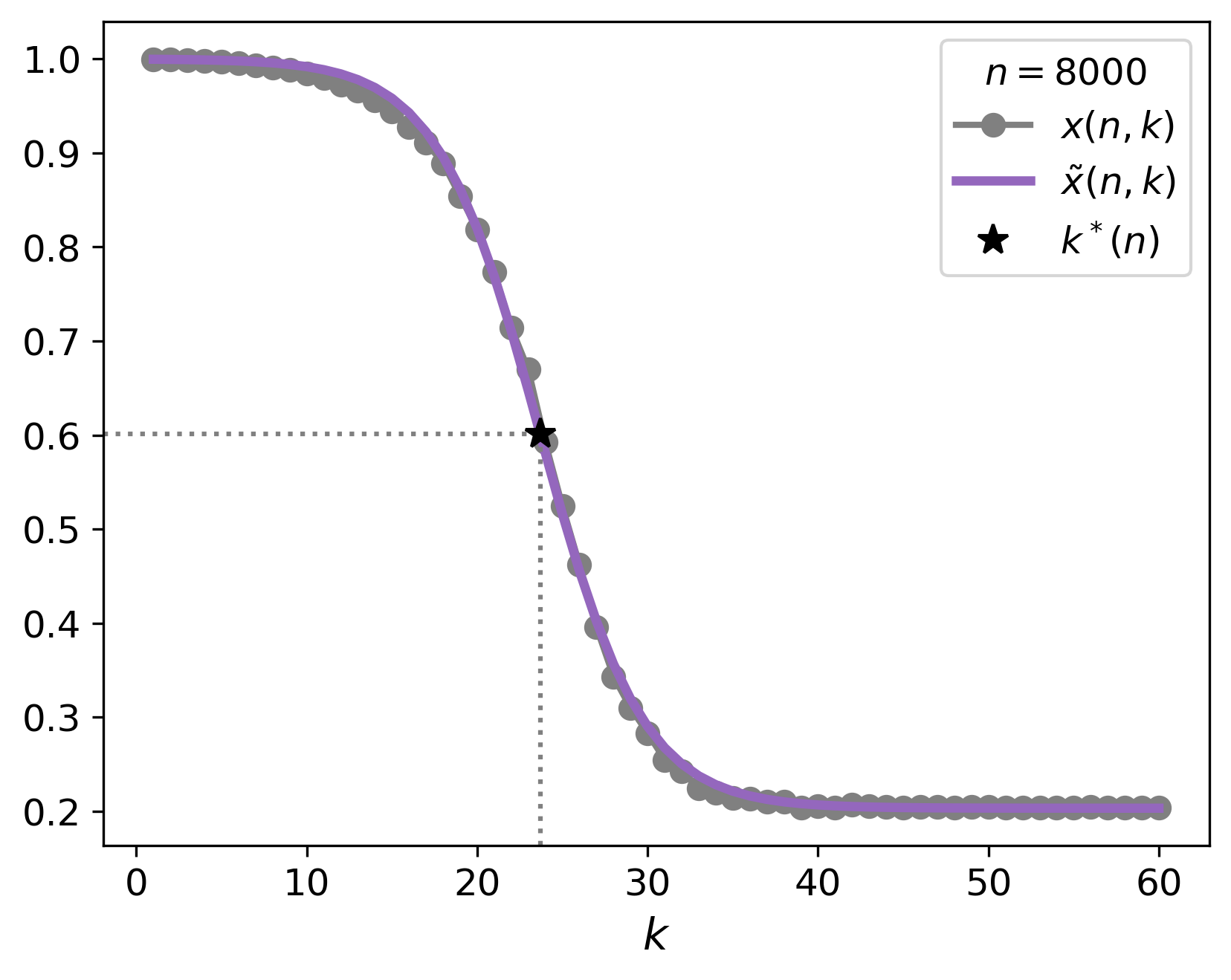}}
    \subfigure[$\alpha = -0.324483$ y $\beta = 26.763353$.]{\includegraphics[width=0.45\linewidth]{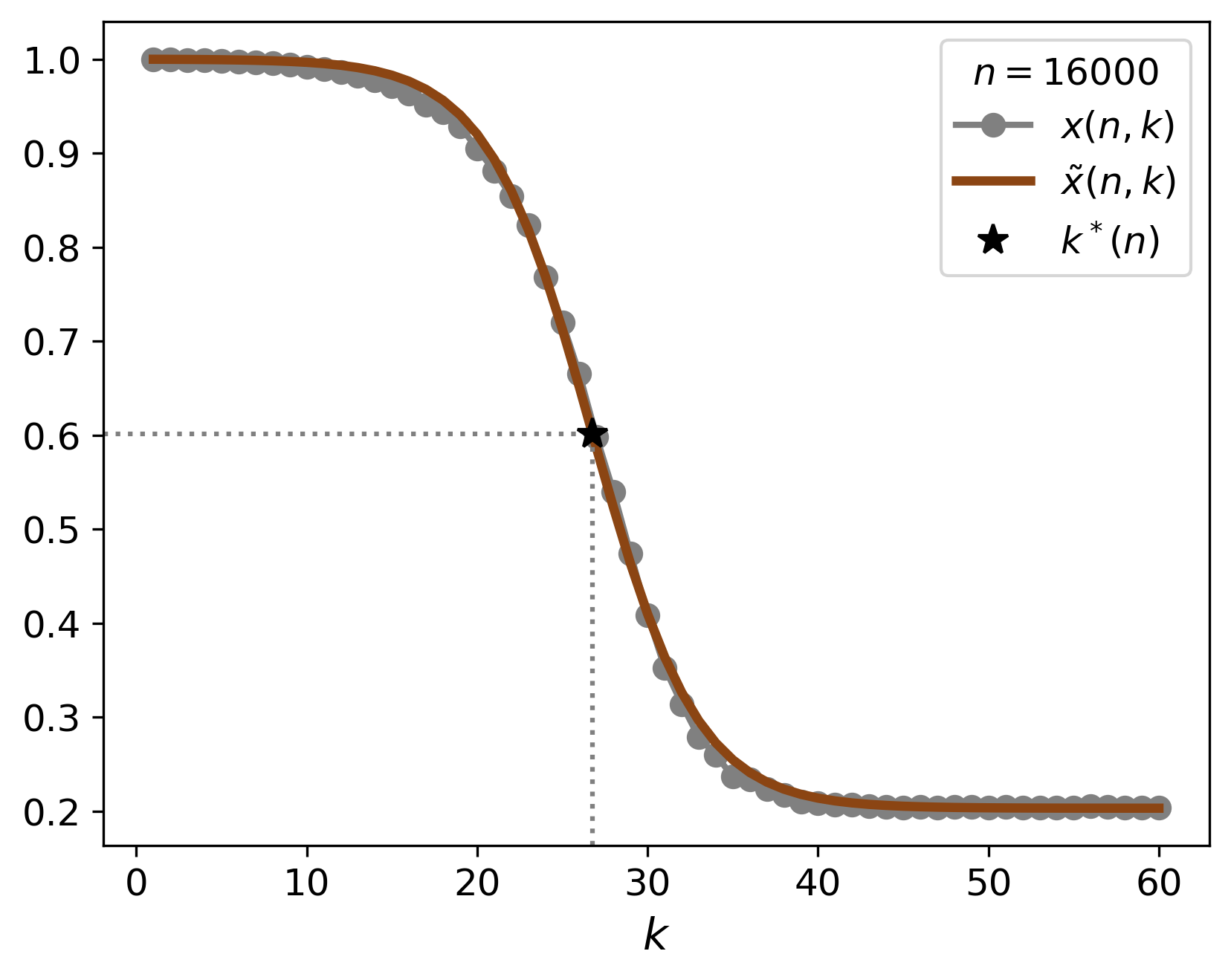}}
    \subfigure[$\alpha = -0.320134$ and $\beta = 29.997113$.]{\includegraphics[width=0.45\linewidth]{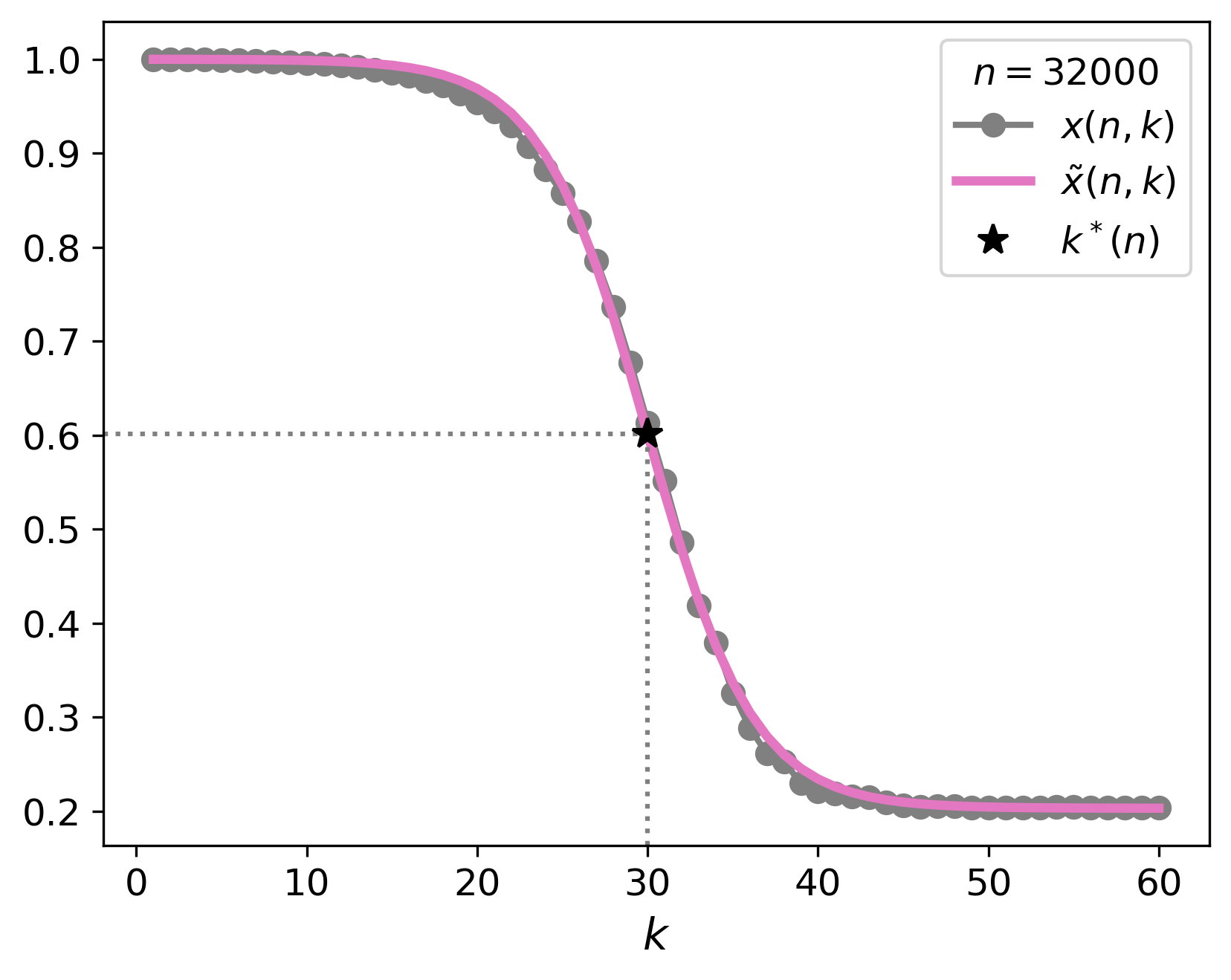}}
    \caption{Fit of the empirical average of the final proportion of ignorants, $x(n,k)$ (gray dots), to a logistic function, $\tilde x(n,k)$ (curves), for $n=1000$. The inflection point of $\tilde{x}(n,k)$, which occurs at $(k^*(n),0.6015935)$, is identified with a black star.}
    \label{fig:fit}
\end{figure*}

At the intersection of the two dotted lines in Figure \ref{fig:fit}, we obtain the point $(k^*,z_\infty/2+x_\infty)$ at which the inflection point of $\tilde{x}(n,k)$ occurs, where $z_\infty/2+x_\infty\approx 0.6015935$ and $k^*:=k^*(n)$. Because of the symmetry and behavior of the function, this point is of interest for the analysis. It is concluded that $k^*=\beta$, because substituting $\tilde x(n,k^*)=z_\infty/2+x_\infty$ in Equation \eqref{eq:Logistic} gives
\begin{equation}
    \dfrac{1}{2}=\dfrac{1}{1+e^{-\alpha(k^*-\beta)}},
\end{equation}
and from here we have $-\alpha(k^*-\beta)=0$ and $\alpha \neq 0$.  By representing the values of $k^*$ with respect to $n$, at logarithmic scale, we obtain that $k^*(n)\approx \beta_0\log n + \beta_1,$ for given constants $\beta_0$ and $\beta_1$, see Figure \ref{fig:N_kstar}(a). Thus, by a straight calculation, we conclude that $\beta_0 =  4.435449$ and $\beta_1 = -16.173315$. See the final approximation in Figure \ref{fig:N_kstar}(b). Since $k^*=\beta$, and $z_\infty = 1-x_\infty$, substituting into the equation \eqref{eq:Logistic} yields

\begin{equation}
\tilde x(k,n)=\dfrac{1-x_\infty}{1+e^{-\alpha(k-(\beta_0\log n + \beta_1))}}+x_\infty.
    \label{eq:LogistickN}
\end{equation}

\begin{figure*}[h!]
    \center
    \subfigure[]{\includegraphics[width=0.47\linewidth]{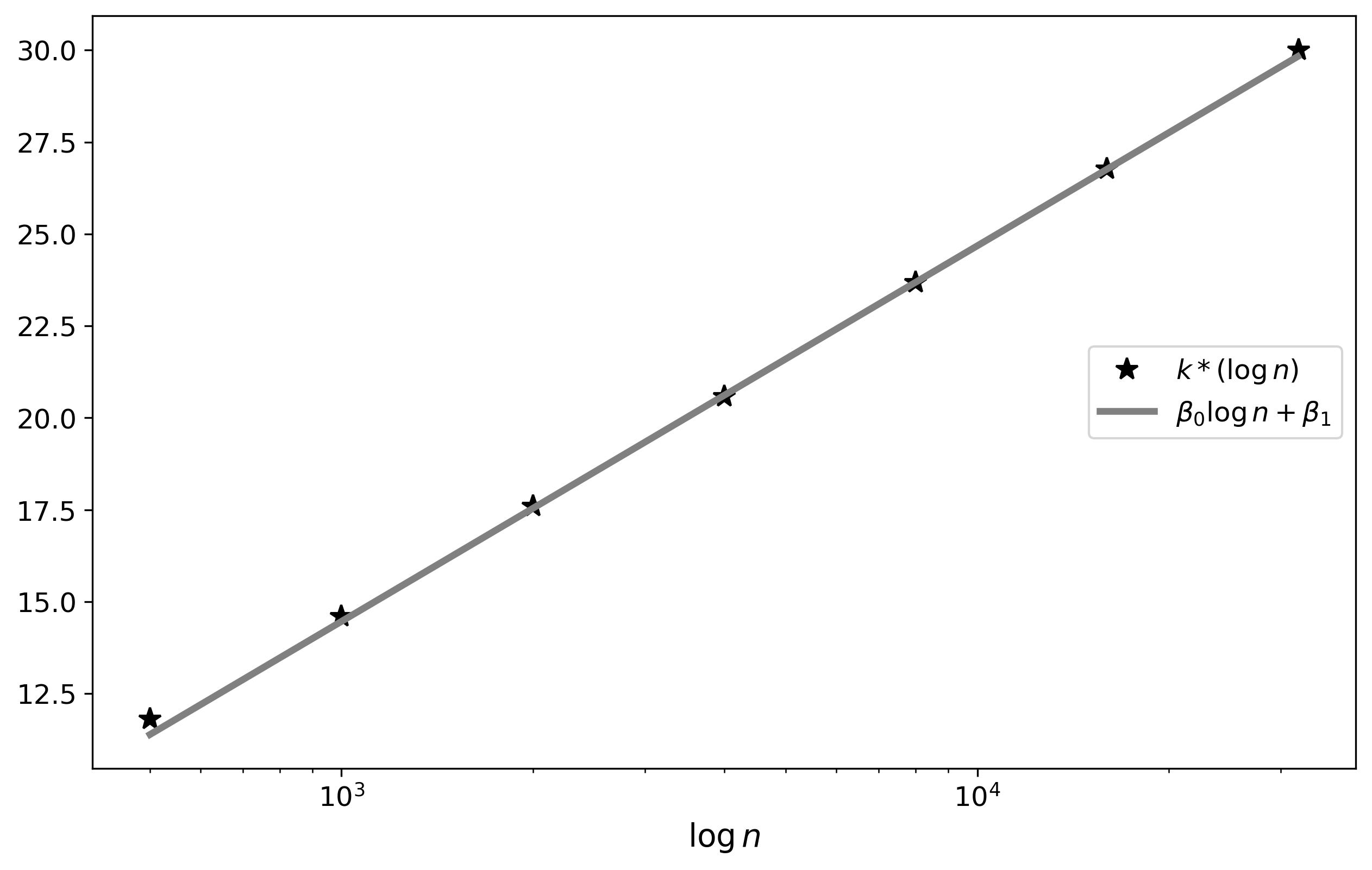}}
    \subfigure[]{\includegraphics[width=0.47\linewidth]{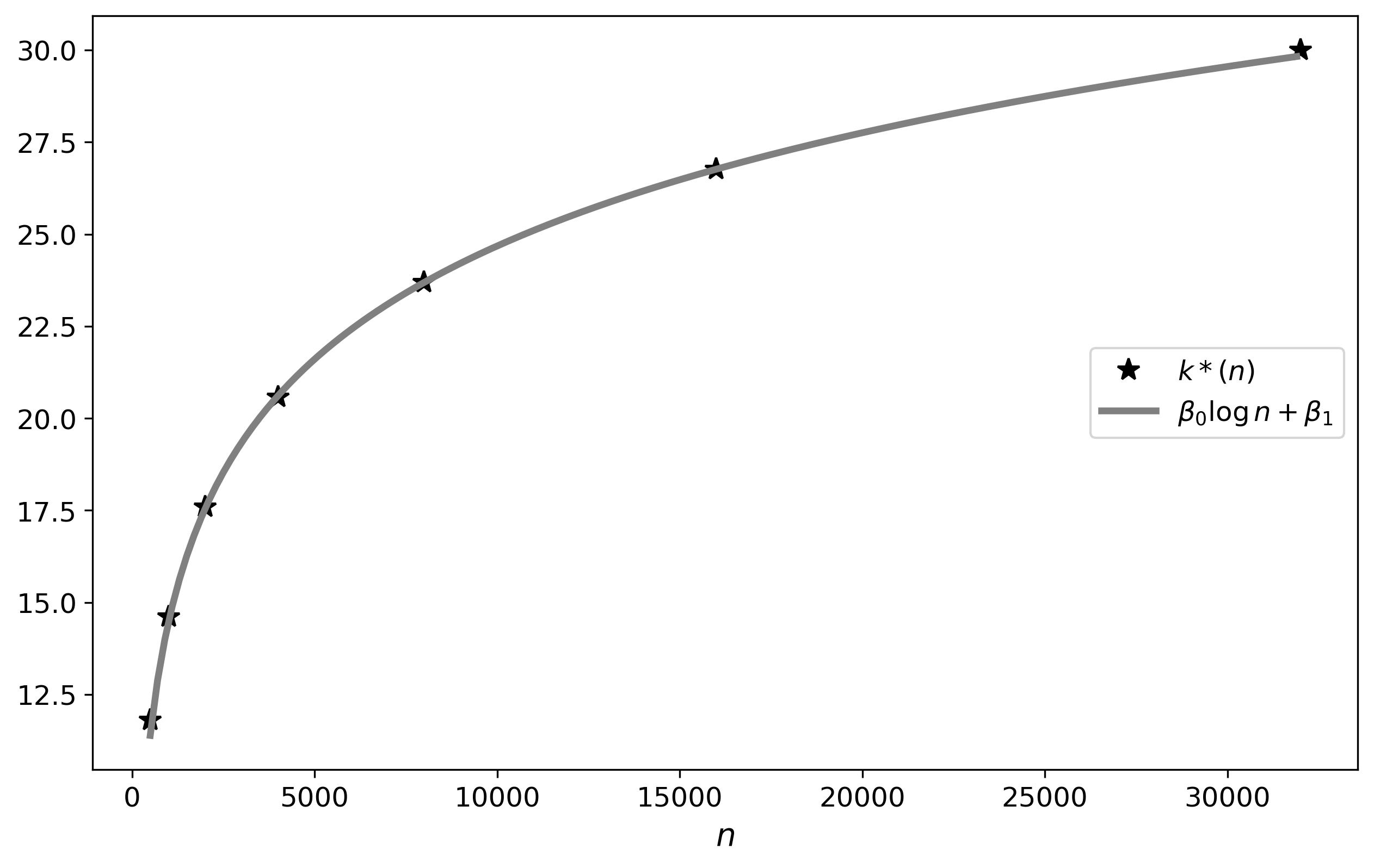}}
    \caption{(a) Representation of the values of $k^*$ with respect to $n$ at logarithmic scale and identification of its approximating line. (b) Illustration of the approximation between $k^*(n)$ and $\beta_0\log n+\beta_1 $, where $\beta_0 =  4.435449$ and $\beta_1 = -16.173315$ are obtained from the approximation in (a), for $n\in\{500, 1000, 2000, 4000, 8000, 16000, 32000\}$.}
    \label{fig:N_kstar}
\end{figure*}

\subsection{A critical degree for which an interesting behavior emerges}

The preceding sections are particularly interesting when one realizes the existence of a critical degree value for each n, after which the behavior concerning the proportion of nodes reached by the information is analogous to that observed in homogeneously mixed populations, that is, the complete graph. In other words, the remaining proportion of ignorants is approximately $0.203$, provided that $k$ is greater than this critical value. Given the numerical nature of our approach, it is subject to errors. Consequently, rather than identifying the precise critical degree value, our investigation focused on the order of magnitude as a function of $n$. Now fix $\epsilon>0$ small enough and let us find for $n$ large enough a value $k_{\epsilon}(n)$ such that

\begin{equation}\label{eq:diff}
|x(k,n) - x_{\infty}| < \epsilon \,\,\, \text{ for any } k\geq k_{\epsilon}(n).
\end{equation}

That is, for $\epsilon>0$, we define for $n$ large enough:

\begin{equation}\label{eq:criticalk}
k_{\epsilon}(n):=\min\{k\geq 1: |x(k,n) - x_{\infty}| < \epsilon\}.
\end{equation}

Note that \eqref{eq:criticalk} is well defined because $x(k,n)$ is non-increasing as a function of $k$. Moreover, note that \eqref{eq:diff} is equivalent, by \eqref{eq:LogistickN} and the fact that $x(k,n)\approx \tilde x(k,n)$ for $k\in\{1,2,\ldots\}$, to

$$
\dfrac{1-x_\infty}{\epsilon} -1 < e^{-\alpha(n)(k-(\beta_0\log n +\beta_1) )}.
$$

This, in turns, is equivalent to

$$
\log\left(\dfrac{1-x_\infty}{\epsilon} -1\right) < -\alpha(n)(k-(\beta_0\log n + \beta_1)).
$$

In other words, since $\alpha(n) <0$ and $\beta_1 <0$, we can rewrite the previous inequality to obtain:

$$
k> |\alpha(n)|^{-1}\log\left(\dfrac{1-x_\infty}{\epsilon} -1\right) + \beta_0\log n + \beta_1=:g_{\epsilon}(n).
$$

From the previous steps, we can consider $k_{\epsilon}(n)=\lceil g_{\epsilon}(n)\rceil$. Now, from our simulations we verify that $\alpha(n)\approx \alpha_0\{\log(n+\alpha_1)+\alpha_2\}^{-1}+\alpha_3$, where $\alpha_0 = -0.120584$, $\alpha_1 = 531.528068$, $\alpha_2 = -5.756993$ and $\alpha_3 = -0.294318$, see Figure \ref{fig:alpha_fit}. This in turn implies that $|\alpha(n)|^{-1} = o(\log n)$, and thus $g_{\epsilon}(n)=(o(1)+\beta_0)\log n=O(\log n)$. In particular, $k_{\epsilon}(n)=O(\log n)$.

\begin{figure*}[h!]
    \centering
    \includegraphics[width=0.6\linewidth]{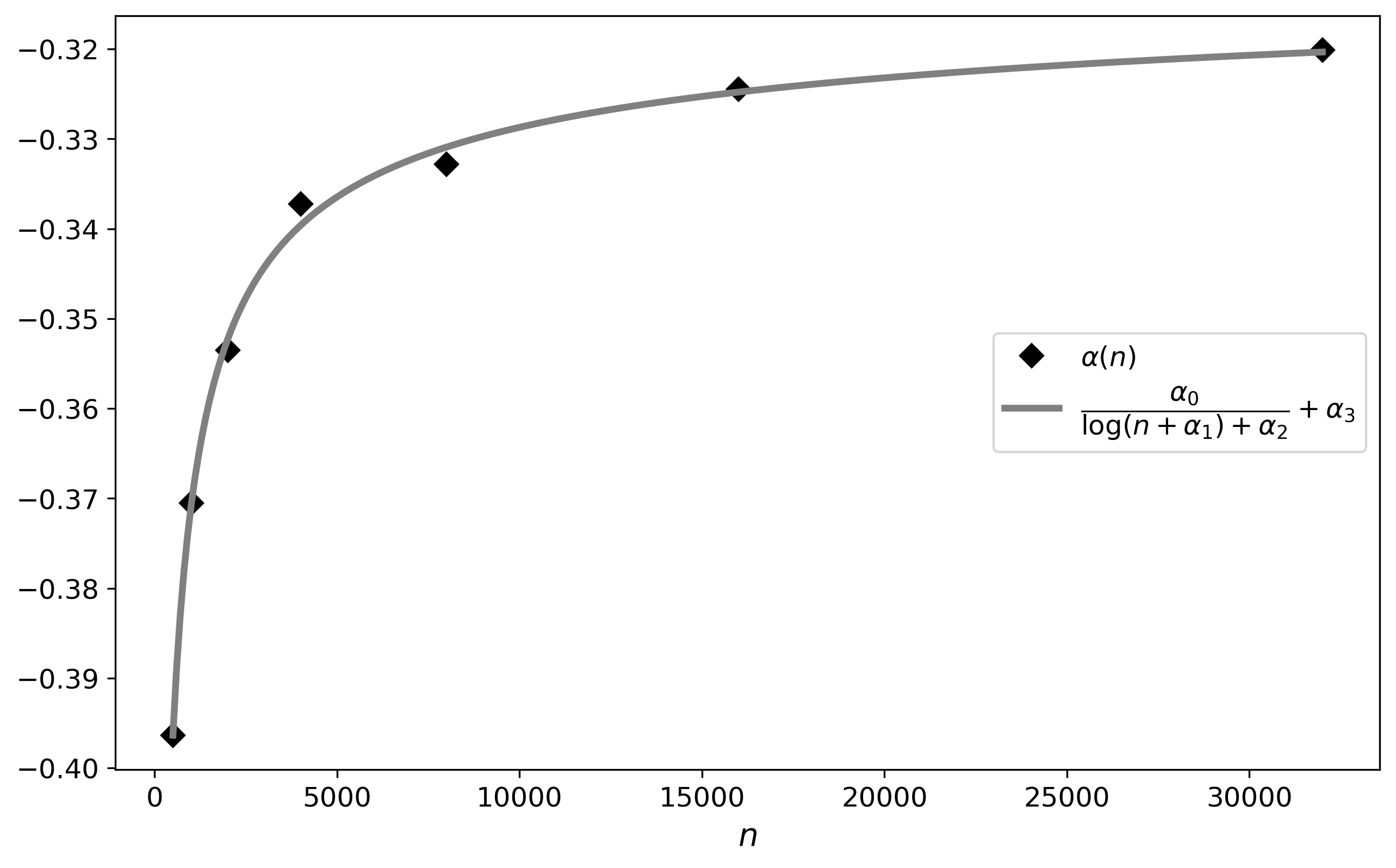}
    \caption{Approximation for $\alpha(n)$ through $\alpha_0\{\log(n+\alpha_1)+\alpha_2\}^{-1}+\alpha_3$, where $\alpha_0 = -0.120584$, $\alpha_1 = 531.528068$, $\alpha_2 = -5.756993$ and $\alpha_3 = -0.294318$, for $n\in\{500, 1000, 2000, 4000, 8000, 16000, 32000\}$.}
    \label{fig:alpha_fit}
\end{figure*}

\section{Conclusion}

We consider the Maki-Thompson rumor model on a $2k$-regular ring lattice with $n$ nodes. After discussing theoretical results for two extreme cases--the complete graph ($k = \lfloor n/2 \rfloor$) and the cycle graph ($k = 1$)--we carry out a numerical analysis using Monte Carlo simulations for general values of $k$. We focus on how the behavior of the process depends on $k$ as a function of $n$. We find that the final proportion of ignorant individuals is approximately the same as that in the complete graph, provided that $k(n)$ is on the order of $\log n$. The significance of studying the ring lattice lies in the fact that this type of graph forms the basis for several small-world network models, such as the Watts--Strogatz network and its Newman--Watts variant. For these networks, it is well known that when a rumor spreads, the process undergoes a transition between localized and propagation at a specific level of network randomness \cite{EAP,zanette,zanette02}. These findings highlight the strong influence of shortcuts on the spread of information in such networks. {\color{black}We emphasize that although those works also suggest that the difference between the regimes decreases as $k$ grows, their analyses were restricted to small or fixed values of $k$. Therefore, as $n$ increases, the role of $k$ is not taken into account in the same way as in our work.} Our approach allows us to show that even without shortcuts, the model can display behavior, in terms of the proportion of informed nodes, similar to what is observed in homogeneously mixed populations. {\color{black}Moreover, since we identify the value of $k$ as a function of $n$ for which this behavior emerges, showing that it scales as $\log n$, the same behavior can be expected in small-world networks even for small values of the randomness parameter. Note that the condition $k \gg \log n$ was imposed by \cite{watts} to guarantee that the resulting random network is connected. In that work, the authors also suggest examples of actual data that can be modeled by a small-world network, such as the collaboration graph of film actors.}  Since the Maki–Thompson model is one of the earliest mathematical models proposed for rumor spreading and has inspired extensive research on information transmission in graphs, our analysis can be adapted to study more general models of social contagion.

\section*{CRediT authorship contribution statement}

Ana C. D\'iaz Bacca: Conceptualization, Methodology, Formal analysis, Investigation, Writing - Review $\&$ Editing. Pablo M. Rodriguez: Conceptualization, Methodology, Formal analysis, Investigation, Writing - Original Draft, Writing - Review $\&$ Editing, Supervision, Funding acquisition. Catalina M. Rúa-Alvarez: Conceptualization, Methodology, Formal analysis, Investigation, Writing - Original Draft, Writing - Review $\&$ Editing, Supervision.

\section*{Declaration of competing interest}
We have nothing to declare.

\section*{Acknowledgements}
This work was partially supported by the Conselho Nacional de Desenvolvimento Científico e Tecnológico - CNPq (Grant 316121/2023-1), the Funda\c{c}\~ao de Amparo \`a Ci\^encia e Tecnologia do Estado de Pernambuco - FACEPE (Grants APQ-1341-1.02/22 and IBPG-1967-1.02/21), and Fundação de
Amparo à Pesquisa do Estado de São Paulo - FAPESP (Grant 23/13453-5). C. M. Rúa-Alvarez was partially supported by Vicerrectoría de Investigaciones e Interacción Social at Universidad de Nariño.

%%%%%% REFERENCIAS


\begin{thebibliography}{99}

%Van der Geer, J., Hanraads, J. A. J., &amp; Lupton, R. A. (2010). The art of writing a scientific article. Journal of Scientific Communications, 163, 51-59.

\bibitem{EAP}
Agliari, E., Pachon, A., Rodriguez, P. M. $\&$ Tavani, F. (2017). Phase transition for the Maki-Thompson rumor model on a small-world network, Journal of Statistical Physics, 169, 846-875. % n.4,


%\bibitem{arruda3}
%de Arruda, G. F., Rodrigues, F. A., Rodriguez, P. M., Cozzo, E. and Moreno, Y., A General Markov Chain Approach for Disease and Rumor Spreading in Complex Networks, J. Complex Netw. 6(2) (2018), 215-242.

%\bibitem{BranchingProcesses}
%Athreya, K. B. $\&$ Ney, P. E. (1972). {\it Branching processes}. Springer-Verlag, Heidelberg. %Die Grundlehren der mathematischen Wissenschaften, Band 196.
%
	
%\bibitem{belen08}
%Belen, S., 2008. The behaviour of stochastic rumours. Ph.D. thesis, School of Mathematical Sciences, University of Adelaide, Australia.
%Available at \url{http://hdl.handle.net/2440/49472}.


\bibitem{raey}
Coletti, C. F., Rodr\'iguez, P. M. $\&$ Schinazi, R. B. (2012). A Spatial Stochastic Model for rumor Transmission, Journal of Statistical Physics, 147, 375-381.

\bibitem{dg}
Daley, D. J. $\&$ Gani, J. (1999). {\it Epidemic Modelling: an Introduction}. Cambridge University Press, Cambridge.   
%
\bibitem{daley_nature}
Daley, D. J. $\&$ Kendall, D. G. (1964). Epidemics and Rumours, Nature, 204, 1118.

\bibitem{kendall}
Daley, D. J. $\&$ Kendall, D. G. (1965). Stochastic rumours,  Journal of the Institute of Mathematics and its Applications, 1, 42-55.


%\bibitem{gani-Env2000}
%Gani, J. (2000). The Maki-Thompson rumour model: a detailed analysis, Environmental Modelling $\&$ Software, 15, 721-725.

\bibitem{grimmett}
Grimmett, J. (2018). {\it Probability on Graphs: Random Processes on Graphs and Lattices}. Institute of Mathematical Statistics Textbooks. Cambridge University Press.

\bibitem{juniormaki}
Junior, V. V., Rodriguez, P. M. $\&$ Speroto, A. (2020). The Maki-Thompson rumor model on infinite Cayley trees, Journal of Statistical Physics, 181 n.4, 1204-1217.


\bibitem{juniorstochastic}
Junior, V. V., Rodriguez, P. M. $\&$ Speroto, A. (2021). Stochastic rumors on random trees, Journal of Statistical Mechanics: Theory and Experiment, 2021 n. 12, 123403.

\bibitem{kurtz}
Kurtz, T.G., Lebensztayn, E., Leichsenring, A.R. $\&$ Machado, F.P. (2008). Limit theorems for an epidemic model on the complete graph, ALEA. Latin American Journal of Probability and Mathematical Statistics, 4, 45e55.
%\bibitem{Lebensztayn-JMAA2015}
%Lebensztayn, E., A large deviations principle for the Maki-Thompson rumour model, J. Math. Analysis and Applications 432 (2015), 142-155.
%
%\bibitem{pittel-JAP1987}
%Pittel, B., On a Daley-Kendal model of rumours, J. Appl. Probab. 27 (1987).
%
\bibitem{lebensztayn/machado/rodriguez/2011a}
Lebensztayn, E., Machado, F. P. $\&$ Rodr\'iguez, P. M. (2011). On the behaviour of a rumour process with random stifling, Environmental Modelling $\&$ Software, 26, 517-522.
%
\bibitem{lebensztayn/machado/rodriguez/2011b}
Lebensztayn, E., Machado, F. $\&$ Rodr\'iguez, P. M. (2011). Limit Theorems for a General Stochastic Rumour Model, SIAM Journal on Applied Mathematics, 71, 1476-1486.
%
\bibitem{EP}
Lebensztayn, E. $\&$ Rodriguez, P. M. (2013). A connection between a system of random walks and rumor transmission, Physica A, 392, 5793-5800.

{\color{black}
\bibitem{li}
Li, B. $\&$ Zhu, L. (2024). Turing instability analysis of a reaction–diffusion system for rumor propagation in continuous space and complex networks, Information Processing $\&$ Management, 61(3), 103621. %https://doi.org/10.1016/j.ipm.2023.103621.
}

%
\bibitem{MT}
Maki, D. P. $\&$ Thompson, M. (1973). {\it Mathematical Models and Applications. With Emphasis on the Social, Life, and Management Sciences}. Prentice-Hall, Englewood Cliffs, New Jersey.
%
\bibitem{MNP-PRE2004}
Moreno, Y., Nekovee, M., $\&$ Pacheco, A. F. (2004). Dynamics of rumour spreading in complex networks, Phys. Rev. E, 69, 066130.

%
\bibitem{moreno-PhysA2007}
Nekovee, M., Moreno, Y., Bianconi, G., $\&$ Marsili, M. (2007). Theory of rumour spreading in complex social neworks, Physica A, 374, 457-470.
%
\bibitem{puerres}
Puerres, J., Junior, V.V., $\&$ Rodriguez, P.M. (2025). Critical thresholds in stochastic rumors on trees, Chaos Solit. Fractals, 201, 117373.

%
\bibitem{rada}
Rada, A., Coletti, C., Lebensztayn, E. $\&$ Rodriguez, P. M. (2021). The Role of Multiple Repetitions on the Size of a Rumor, SIAM Journal on Applied Dynamical Systems, 20(3), 1209-1231.
%
\bibitem{resnick}
Resnick, S. I. (2005). {\it A Probability Path}, Birkhauser Boston, 5th printing.
%

{\color{black}
\bibitem{shi}
Shi, J. $\&$ Zhu, L. (2025). Turing pattern theory on homogeneous and heterogeneous higher-order temporal network system, J. Math. Phys., 66(4), 042706. %https://doi.org/10.1063/5.0211728
}
%

{\color{black}
\bibitem{sha}
Sha, H. $\&$  Zhu, L. (2025). Dynamic analysis of pattern and optimal control research of rumor propagation model on different networks, Information Processing $\&$ Management,
62(3), 104016.}% https://doi.org/10.1016/j.ipm.2024.104016.



%%
\bibitem{sudbury}
Sudbury, A. (1985). The proportion of the population never hearing a rumour. Journal of Applied Probability, 22, 443–446.

\bibitem{watson}
Watson, R. (1988). On the size of a rumour, Stochastic Processes and their Applications, 27, 141–149.

{\color{black}
\bibitem{watts}
Watts, D. J. $\&$  Strogatz, S. H. (1998). Collective dynamics of `small-world' networks, Nature 393(6684), 440–442.}

{\color{black}
\bibitem{yang}
Yang, T., Zhu, L., Shen, S. $\&$ He, L. (2025). Pattern dynamics analysis and parameter identification of spatiotemporal infectious disease models on complex networks, Mathematical Biosciences, 387, 109502.% https://doi.org/10.1016/j.mbs.2025.109502.
}
%
\bibitem{zanette}
Zanette, D. H. (2001). Critical behavior of propagation on small-world networks, Phys. Rev. E, 64, (R)050901.
%
%
\bibitem{zanette02}
Zanette, D. H. (2002). Dynamics of rumour propagation on small-world networks, Phys. Rev. Lett., 65(4), 041908.
%

{\color{black}
\bibitem{zhu-ding}
Zhu, L., Ding, Y. $\&$ Shen, S. (2025). Green behavior propagation analysis based on statistical theory and intelligent algorithm in data-driven environment, Mathematical Biosciences, 379, 109340.% https://doi.org/10.1016/j.mbs.2024.109340.


\bibitem{zhu}
Zhu, L. $\&$  Zheng, T. (2025). Pattern dynamics analysis and application of West Nile virus spatiotemporal models based on higher-order network topology, Bull Math Biol 87, 121.}% https://doi.org/10.1007/s11538-025-01501-6}
%\bibitem{arruda2}	
%de Arruda, G. F., Lebensztayn, E., Rodrigues,  F. A. and Rodriguez, P. M., A process of rumour scotching on finite populations, R. Soc. open sci. 2 (2015) 150240.


\end{thebibliography}
\end{document}